%% file: ms.tex
\RequirePackage[l2tabu,orthodox]{nag}
\documentclass
[letterpaper,12pt,]
{article}

\usepackage{etex}
\usepackage{verbatim}
\usepackage{xspace,enumerate}
\usepackage[dvipsnames]{xcolor}
\usepackage[T1]{fontenc}
\usepackage[full]{textcomp}
\usepackage[american]{babel}
\usepackage{mathtools}

\usepackage{amsthm}
\usepackage{thmtools}
\usepackage{thm-restate}

\usepackage[
letterpaper,
top=1in,
bottom=1in,
left=1in,
right=1in]{geometry}
\usepackage{mdframed}

\usepackage{newpxtext} %
\usepackage{textcomp} %
\usepackage[varg,bigdelims]{newpxmath}
\usepackage[scr=rsfso]{mathalfa}%
\usepackage{bm} %
\let\mathbb\varmathbb
\usepackage{microtype}
\usepackage[pagebackref,colorlinks=true,urlcolor=blue,linkcolor=blue,citecolor=OliveGreen]{hyperref}
\usepackage[capitalise,nameinlink]{cleveref}
\crefname{lemma}{Lemma}{Lemmas}
\crefname{fact}{Fact}{Facts}
\crefname{theorem}{Theorem}{Theorems}
\crefname{corollary}{Corollary}{Corollaries}
\crefname{claim}{Claim}{Claims}
\crefname{example}{Example}{Examples}
\crefname{algorithm}{Algorithm}{Algorithms}
\crefname{problem}{Problem}{Problems}
\crefname{definition}{Definition}{Definitions}
\crefname{exercise}{Exercise}{Exercises}
\usepackage{amsthm}

\newtheorem{theorem}{Theorem}[section]
\newtheorem*{theorem*}{Theorem}
\newtheorem{lemma}[theorem]{Lemma}
\newtheorem*{lemma*}{Lemma}
\newtheorem{fact}[theorem]{Fact}
\newtheorem*{fact*}{Fact}

\newtheorem*{proposition*}{Proposition}
\newtheorem{corollary}[theorem]{Corollary}
\newtheorem*{corollary*}{Corollary}

\newtheorem*{hypothesis*}{Hypothesis}

\newtheorem*{conjecture*}{Conjecture}
\theoremstyle{definition}
\newtheorem{definition}[theorem]{Definition}
\newtheorem*{definition*}{Definition}

\newtheorem*{construction*}{Construction}

\newtheorem*{example*}{Example}

\newtheorem*{question*}{Question}
\newtheorem{algorithm}[theorem]{Algorithm}
\newtheorem*{algorithm*}{Algorithm}

\newtheorem*{assumption*}{Assumption}

\newtheorem*{problem*}{Problem}

\newtheorem*{openquestion*}{Open Question}
\theoremstyle{remark}

\newtheorem*{claim*}{Claim}

\newtheorem*{remark*}{Remark}

\newtheorem*{observation*}{Observation}
\usepackage{paralist}
\let\originalleft\left
\let\originalright\right
\renewcommand{\left}{\mathopen{}\mathclose\bgroup\originalleft}
\renewcommand{\right}{\aftergroup\egroup\originalright}
\usepackage{turnstile}
\usepackage{mdframed}
\usepackage{tikz}
\usepackage{caption}
\DeclareCaptionType{Algorithm}
\usepackage{newfloat}
\usepackage{xparse}
\usepackage{amsthm} %
\makeatletter
\let\latexparagraph\paragraph
\RenewDocumentCommand{\paragraph}{som}{%
  \IfBooleanTF{#1}
    {\latexparagraph*{#3}}
    {\IfNoValueTF{#2}
       {\latexparagraph{\maybe@addperiod{#3}}}
       {\latexparagraph[#2]{\maybe@addperiod{#3}}}%
  }%
}
\newcommand{\maybe@addperiod}[1]{%
  #1\@addpunct{.}%
}
\makeatother

\usepackage{boxedminipage}

\newcommand{\Esymb}{\mathbb{E}}
\newcommand{\Psymb}{\mathbb{P}}

\DeclareMathOperator*{\E}{\Esymb}

\DeclareMathOperator*{\ProbOp}{\Psymb}
\renewcommand{\Pr}{\ProbOp}

\newcommand\bdot\bullet

\DeclareMathOperator{\Tr}{Tr}

\newcommand{\Hoelder}{H\"{o}lder\xspace}

\newcommand{\N}{\mathbb N}
\newcommand{\R}{\mathbb R}

\newcommand{\cA}{\mathcal A}
\newcommand{\cB}{\mathcal B}

\newcommand{\cN}{\mathcal N}

\let\epsilon=\varepsilon
\numberwithin{equation}{section}
\newcommand\MYcurrentlabel{xxx}
\newcommand{\MYstore}[2]{%
  \global\expandafter \def \csname MYMEMORY #1 \endcsname{#2}%
}
\newcommand{\MYload}[1]{%
  \csname MYMEMORY #1 \endcsname%
}
\newcommand{\MYnewlabel}[1]{%
  \renewcommand\MYcurrentlabel{#1}%
  \MYoldlabel{#1}%
}
\newcommand{\MYdummylabel}[1]{}
\newcommand{\torestate}[1]{%
  \let\MYoldlabel\label%
  \let\label\MYnewlabel%
  #1%
  \MYstore{\MYcurrentlabel}{#1}%
  \let\label\MYoldlabel%
}
\newcommand{\restatetheorem}[1]{%
  \let\MYoldlabel\label
  \let\label\MYdummylabel
  \begin{theorem*}[Restatement of \cref{#1}]
    \MYload{#1}
  \end{theorem*}
  \let\label\MYoldlabel
}
\newcommand{\restatelemma}[1]{%
  \let\MYoldlabel\label
  \let\label\MYdummylabel
  \begin{lemma*}[Restatement of \cref{#1}]
    \MYload{#1}
  \end{lemma*}
  \let\label\MYoldlabel
}
\newcommand{\restateprop}[1]{%
  \let\MYoldlabel\label
  \let\label\MYdummylabel
  \begin{proposition*}[Restatement of \cref{#1}]
    \MYload{#1}
  \end{proposition*}
  \let\label\MYoldlabel
}
\newcommand{\restatefact}[1]{%
  \let\MYoldlabel\label
  \let\label\MYdummylabel
  \begin{fact*}[Restatement of \cref{#1}]
    \MYload{#1}
  \end{fact*}
  \let\label\MYoldlabel
}
\newcommand{\restate}[1]{%
  \let\MYoldlabel\label
  \let\label\MYdummylabel
  \MYload{#1}
  \let\label\MYoldlabel
}

\allowdisplaybreaks
\DeclarePairedDelimiterX{\infdivx}[2]{(}{)}{#1\;\delimsize\|\;#2}
\newcommand{\kldiv}{D_{KL}\infdivx}

\newcommand*{\Lowner}{L\"owner\xspace}

\title{
  Outlier-robust Mean Estimation near the Breakdown Point via Sum-of-Squares
  \thanks{This project has received funding from the European Research Council (ERC) under the European Union’s Horizon 2020 research and innovation programme (grant agreement No 815464)}
}

\renewcommand{\epsilon}{\ensuremath\varepsilon}

\renewcommand{\phi}{\ensuremath{\varphi}}
\newcommand{\xs}{{x}^*}
\newcommand{\Ss}{{\Sigma}^*}
\newcommand{\mus}{{\mu}^*}
\newcommand{\sn}{\frac{1}{n}\sum\limits_{i=1}^n}
\newcommand{\snj}{\frac{1}{n}\sum\limits_{j=1}^n}

\newcommand{\ip}{\{z_i\}_{i=1}^n}
\newcommand{\uip}{\{ \xs_i \}_{i=1}^n}
\newcommand{\pwi}{\{w_i\}_{i=1}^n}
\newcommand{\pxi}{\{x_i\}_{i=1}^n}

\newcommand{\vnorm}[1]{\left\lVert#1\right\rVert}
\newcommand{\one}{\mathbb{1}}
\newcommand{\Ex}[1]{\E\left[#1\right]}

\newcommand{\err}{\vnorm{\mu - \mu^*}}
\newenvironment{algorithmbox}{\begin{mdframed}[nobreak=true]\begin{algorithm}}{\end{algorithm}\end{mdframed}}
\newcommand{\del}{\mu - \mu^*}

\newcommand{\pE}{\widetilde{\mathbb{E}}}

\newcommand{\filter}{filter }

\newcommand{\D}{\mathcal{D}}
\newcommand{\Pp}{\mathcal{P}}
\newcommand{\ot}{\widetilde{O}(\varepsilon)}
\author{
  Hongjie Chen\thanks{ETH Z\"urich.}
  \and
  Deepak Narayanan Sridharan\footnotemark[2]
  \and 
  David Steurer\footnotemark[2]\\
}

\date{}
\begin{document}

\pagestyle{empty}

\maketitle
\thispagestyle{empty} %

\begin{abstract}

\input{abstract}

\end{abstract}

\clearpage

\microtypesetup{protrusion=false}
\tableofcontents{}
\microtypesetup{protrusion=true}

\clearpage

\pagestyle{plain}
\setcounter{page}{1}

\input{introduction}
\input{techniques}
\input{notation}

\input{algorithms}

\input{sparse}

\newpage
\phantomsection
\addcontentsline{toc}{section}{References}
\bibliographystyle{amsalpha}
\bibliography{refs}
\newpage

\appendix
\input{appendix}
\end{document}

%% file: abstract.tex
               We revisit the problem of estimating the mean of a high-dimensional distribution in the presence of an $\varepsilon$-fraction of adversarial outliers.
    When $\varepsilon$ is at most some sufficiently small constant, previous works can achieve optimal error rate efficiently \cite{diakonikolas2018robustly, kothari2018robust}. 
    As $\varepsilon$ approaches the breakdown point $\frac{1}{2}$, all previous algorithms incur either sub-optimal error rates or exponential running time.  
    
    In this paper we give a new analysis of the canonical sum-of-squares program introduced in \cite{kothari2018robust} and show that this program efficiently achieves optimal error rate for all $\varepsilon \in[0,\frac{1}{2})$. The key ingredient for our results is a new identifiability proof for robust mean estimation that focuses on the overlap between the distributions instead of their statistical distance as in previous works. We capture this proof within the sum-of-squares proof system, thus obtaining efficient algorithms using the sum-of-squares proofs to algorithms paradigm \cite{raghavendra2018high}.

%% file: introduction.tex
\newcommand{\package}{\emph}

\section{Introduction}

The problem of estimating the mean of a high-dimensional distribution given samples from the distribution is a foundational problem in statistics. In the problem of high-dimensional robust mean estimation, we are instead interested in estimating the mean of a distribution when we observe corrupted samples, making it a natural extension of mean estimation. The main challenge in high-dimensional robust mean estimation stems from the fact that efficient estimators that work well in one-dimension often do not generalize well to high dimensions.  In particular, the error achieved by estimators such as the geometric median and coordinate-wise median scales with the dimension \cite{diakonikolas2023algorithmic} while information-theoretically this dependence on dimension is unnecessary. On the other hand estimators that achieved dimension-independent error such as the Tukey Median \cite{tukey1975mathematics} were hard to compute \cite{johnson1978densest}. 

The breakthrough works of \cite{diakonikolas2019robust, lai2016agnostic} designed the first efficient algorithms for high-dimensional robust mean estimation for Gaussians that achieved near-optimal error. Since then, there has been a plethora of work that has focused on different aspects on the problem of high-dimensional robust mean estimation under different distributional assumptions. Some of these aspects include designing faster and more practical algorithms \cite{diakonikolas2017being, cheng2019high}, designing algorithms for specific corruption models \cite{diakonikolas2018robustly} and exploiting structural properties such as sparsity \cite{balakrishnan2017computationally}. 
Much of the work in robust mean estimation is in the following setting: The algorithm receives as input a corrupted set of samples, a corruption parameter $\varepsilon$, and other information about the underlying distribution. Typically, an $\varepsilon$ fraction of the samples is corrupted by an adversary with knowledge about the underlying distribution and consequently the true mean. The key distinction between different corruption models arises in how these $\varepsilon$ fraction of the samples are corrupted. In the strong contamination model an adversary has full ability to add or remove an $\varepsilon$ fraction of the samples, whereas in the additive contamination model, the adversary is only allowed to add corrupted samples and not \emph{remove} uncorrupted samples. This distinction is critical as different models of corruption can lead to intriguing statistical-computational trade-offs \cite{diakonikolas2017statistical}. In this paper, we consider the strong contamination model and define an $\varepsilon$-corruption as follows.

 \begin{definition}[$\varepsilon$-corruption]\label{def:epscor}
     Let $\D$ be a distribution on $\R^d$ with mean $\mu^*$. Let $\uip \overset{\mathrm{iid}}{\sim} \D$. An adversary with unbounded computational power, knowledge of $\D$, $\mu^*$, and the algorithm inspects the samples $\uip$ and arbitrarily modifies at most an $\varepsilon$ fraction of the samples. The algorithm receives  $\ip$ with the promise that at least $n \cdot (1-\varepsilon)$ samples of $\ip$ agrees with $\uip$.
 \end{definition}

The corruption fraction $\varepsilon$ is also directly related to the \emph{breakdown point} of the estimator. Informally, the breakdown point of an estimator is the maximum $\varepsilon$ that an estimator tolerates before the estimation error becomes unbounded. We note that the type of estimators that we consider have maximum breakdown point of $\frac{1}{2}$ \cite{rousseeuw2005robust, zhu2022robust}. The maximum breakdown is critical as it marks the threshold of statistical identifiability: If $\varepsilon$ is $\frac{1}{2}$, an adversary can simulate the corrupted samples to appear from another distribution from the same family as the true underlying distribution. An algorithm will thus receive as input a balanced mixture and it becomes information-theoretically impossible to {uniquely} identify the mean.

Efficient estimators for robust mean estimation typically work in the regime where $\varepsilon$ is a sufficiently small constant bounded away from $\frac{1}{2}$. Many of these algorithms achieve optimal/near-optimal error up to constants when $\varepsilon$ is sufficiently small \cite{diakonikolas2019robust, diakonikolas2017being,cheng2019high}. While recent work has also focused on improving the breakdown point of estimators and even achieve optimal breakdown point in some cases \cite{dalalyan2022all, zhu2022robust, hopkins2020robust}, these estimators do not achieve optimal error as $\varepsilon $ approaches the breakdown point of  $\frac{1}{2}$\footnote{We note here that there case where $\varepsilon \geqslant \frac{1}{2}$ is studied in the so called list-decodable setting \cite{charikar2017learning, steinhardt2018robust}, where the goal is to output a list of vectors such that one of them is a good estimate for the mean.}. 

Two key methods have emerged as the workhorse behind significant recent developments in robust mean estimation: the filter approach \cite{diakonikolas2017being, diakonikolas2019robust} that iteratively removes and down-weighs outliers, and the Sum-of-Squares (SoS) approach \cite{kothari2018robust,hopkins2018mixture}, based on polynomial optimization.  Some key recent works that use the filter approach include \cite{diakonikolas2017being, diakonikolas2019robust, diakonikolas2018robustly, zhu2022robust} while those that utilize SoS include \cite{kothari2018robust, hopkins2018mixture, kothari2022polynomial, diakonikolas2022robust}.

Algorithms based on the \filter typically assume $\varepsilon$ sufficiently small and are often sub-optimal in terms of the breakdown point. This is partly because algorithms utilizing the filter lack careful analysis and also because these algorithms utilize iterative procedures that succeed only under certain conditions -- conditions that may no longer hold when $\varepsilon$ is large\footnote{ We refer the reader to \cite{li2019lectures} for a related discussion.}. Recently, \cite{hopkins2020robust, zhu2022robust} proposed problem-specific modifications of the standard \filter and demonstrated via careful analysis that algorithms based on this method can indeed achieve optimal breakdown point. However, these algorithms do not achieve optimal error as $\varepsilon \to \frac{1}{2}$. 

Like the \filter method, existing work using SoS is also tailored for small $\varepsilon$ \cite{kothari2018robust, hopkins2018mixture, diakonikolas2022robust} and the behavior of these estimators near the breakdown point is not well understood. This leads us to the natural question:

\begin{itemize}
  \item[]
  \emph{
    Does there exist efficient algorithms achieving optimal error for the full range of corruption for robust mean estimation under different distributional assumptions? 
  }
\end{itemize}

In this paper we answer the above question affirmatively for a large class of distributions. We consider distributions with bounded covariance, certifiably bounded higher moments (See Definition \ref{defn:certifiablebounded}) and Gaussian distributions with known covariance and present efficient estimators based on Sum-of-Squares that achieve optimal breakdown point of $\frac{1}{2}$. For the cases of bounded covariance and certifiably bounded higher moments, our estimators also achieve optimal error (up to constant factors) \emph{efficiently} for the entire range of corruption. For the Gaussian case, our estimator achieves optimal error (up to constant factors) in time that is quasi-polynomial in the inverse of the distance from the breakdown point. We now state our main results and provide corresponding information-theoretic lower bounds. 

\subsection{Our Results}

\subsubsection{Bounded Covariance Distributions}
\begin{theorem}\label{thm:boundedcov}
    Let $\D$ be a distribution on $\R^d$ with mean $\mu^*$ and covariance $\Sigma^*$. 
    Suppose $\Sigma^* \preccurlyeq \sigma^2\cdot I_d$. 
    Let $\uip \overset{\mathrm{iid}}{\sim} \D.$  
    Let $\{z_i\}_{i=1}^n$ be an $\varepsilon$-corruption of $\uip$ where $\varepsilon \in [0, \frac{1}{2})$.
    Let $n = \Omega\left( d \log d \right)$. 
    Then, there is an efficient algorithm based on SoS, running in time $n^{O(1)}$ that takes as input $\ip, \varepsilon, \sigma$, and returns a vector $\hat{\mu}$ such that 
    \begin{align*}
        \Vert \hat{\mu} - \mus \Vert_2 = O\left( \sigma \cdot \sqrt{\frac{\varepsilon}{1-2\varepsilon}} + \sigma \cdot \sqrt{\frac{d}{n}}\right)
    \end{align*}
    with probability 0.99.

\end{theorem}

We remark here that the
$O\left(\sigma\cdot \sqrt{d/n}\right)$ term is necessary even under no corruption. The remaining term of  
$O\left( \sigma \cdot \sqrt{\varepsilon} \cdot (1 - 2\varepsilon)^{-1/2} \right)$ has the information-theoretically optimal dependence on $\varepsilon$ for any $\varepsilon \in [0, \frac{1}{2})$.

We also remark here that existing algorithms, based on the \filter method achieve optimal breakdown point for bounded covariance distributions \cite{zhu2022robust, hopkins2020robust}. Among these algorithms, \cite{zhu2022robust} obtains error $O\left(\sigma \cdot \sqrt{\varepsilon}\cdot(1 - 2\varepsilon)^{-1}\right)$, while \cite{hopkins2020robust} obtains error $O\left(\sigma \cdot \sqrt{\varepsilon}\cdot(1 - 2\varepsilon)^{-3/2}\right)$. When $\varepsilon$ is close to $\frac{1}{2}$, the error obtained by these algorithms is substantially larger than the error attained in Theorem \ref{thm:boundedcov}. It was an open problem of \cite{zhu2022robust} if there was an \emph{efficient} algorithm that achieved error of $O\left(\sigma \cdot \sqrt{\varepsilon} \cdot (1 - 2\varepsilon)^{-1/2}\right)$. Theorem \ref{thm:boundedcov} affirmatively resolves this open problem.

We require $n = \Omega(d\log d)$ to ensure that there exists a sufficiently large subset of samples with bounded empirical covariance, and that the empirical mean of this subset of samples is a good estimate for the true mean of the distribution. We refer the reader to Lemma \ref{lemma:boundedcovconc} for more details. 
\begin{lemma}[Lower Bound for Bounded Covariance]\label{lemma:lbbcov}
    There exists two distributions $\D_1$ and $\D_2$ on $\R$ such that for $\varepsilon \in (0,\frac{1}{2})$
    \begin{enumerate}
        \item $d_{TV}(\D_1, \D_2) \leqslant 2\varepsilon$.
        \item $\D_1$ and $\D_2$ have variance bounded by 1.
        \item $\left| \E_{X \sim \D_1}[X] - \E_{X \sim \D_2}[X] \right| \geqslant \Omega\left(\sqrt{\varepsilon}\cdot (1 - 2\varepsilon)^{-1/2} \right)$.
    \end{enumerate}
\end{lemma}
We are thus the first to demonstrate that an efficient algorithm can achieve the information-theoretically optimal error up to constants for the entire range of corruption for the bounded covariance case. 

\subsubsection{Certifiably bounded higher moment distributions}
We start with the following definition before stating our next result.

\begin{definition}[Distributions with Certifiably Bounded $k^{th}$ moments \cite{hopkins2018mixture}]\label{defn:certifiablebounded} \footnote{In \cite{hopkins2018mixture} the authors refer to this class of distributions as explicitly bounded distributions.}
Let $k \in \mathbb{N}$. A distribution $\D$ over $\R^d$ with mean $\mu$ has certifiably bounded $k^{th}$ moments if for every even $s \leqslant k$
\begin{align*}
    \sststile{s}{u} \E_{X \sim \D}\left[ \langle X - \mu, u \rangle^s \right] \leqslant s^{s/2}\vnorm{u}^s_2.
\end{align*}

\end{definition}
Equivalently, the polynomial $p(u) \coloneqq s^{s/2}\Vert u \Vert_2^{s} - \E_{X \sim \D}\left[ \langle X - \mu, u \rangle^s \right]$ has a sum of squares proof of degree at most $s$ in variables $u$. We refer to this class of distributions as distributions with certifiably bounded higher moments or distributions with certifiably bounded $k^{th}$ moments from hereon for convenience. We remark here that a large class of distributions have higher moments whose boundedness can be certified by an SoS proof \cite{kothari2018robust}.

\begin{theorem}\label{theorem:general}
Let $\D$ be a distribution on $\R^d$  with mean $\mu^*$ and certifiably bounded $k^{th}$ moments in the sense of Definition \ref{defn:certifiablebounded} where $k$ is a power of 2.  
Let $\uip \overset{\mathrm{iid}}{\sim} \D$. Let $\{z_i\}_{i=1}^n$ be an $\varepsilon$-corruption of $\uip$ where $\varepsilon \in [0, \frac{1}{2})$.
Let $n = \Omega\left(d^{O(k)}\right)$.  
Then there is an efficient algorithm based on SoS, running in time $n^{O(k)}$ that takes as input $\ip, \varepsilon, k$, and returns a vector $\hat{\mu}$ such that 
    \begin{align*}
        \Vert \hat{\mu} - \mus \Vert_2 = O\left( \frac{\sqrt{k}\cdot \varepsilon^{1-1/k}}{(1-2\varepsilon)^{1/k}} + \sqrt{\frac{d}{n}}\right)
    \end{align*}
        with probability 0.99.
\end{theorem}

We remark here that the $O\left(\sqrt{d/n}\right)$ term is necessary even without corruption. The remaining term achieves the optimal dependency on $\varepsilon$ for any $\varepsilon \in [0,\frac{1}{2})$. We remark here that the optimal dependence on $\varepsilon$ for sufficiently small $\varepsilon$ is a well-known result \cite{kothari2018robust, hopkins2018mixture}. The optimal error that the above algorithm attains as $\varepsilon \to \frac{1}{2}$ is the first for efficient algorithms.
We require $d^{O(k)}$ many samples for ensuring that the uniform distribution over the uncorrupted samples also satisfies the SoS certifiability as in Definition \ref{defn:certifiablebounded} -- also see Fact 7.6 in \cite{hopkins2018mixture}. 

\begin{lemma}[Lower Bound for Bounded Higher Moments]
    There exists two distributions $\D_1$ and $\D_2$ on $\R$ such that for $\varepsilon \in (0,\frac{1}{2})$
    \begin{enumerate}
        \item $d_{TV}(\D_1, \D_2) \leqslant 2\varepsilon$.
        \item $\D_1$ and $\D_2$ have bounded $k^{th}$ moments for $k \geqslant 2$ in the sense of Definition \ref{defn:certifiablebounded}. 
        \item $\left| \E_{X \sim \D_1}[X] - \E_{X \sim \D_2}[X] \right| \geqslant  \Omega\left(\sqrt{k}\cdot \varepsilon^{1 - 1/k}\cdot{(1 - 2\varepsilon)^{-1/k}}\right)$.
    \end{enumerate}
\end{lemma}
We are thus the first to demonstrate that an efficient algorithm can achieve the information-theoretically optimal error up to constants for the entire range of corruption for distributions with certifiably bounded higher moments.

\subsubsection{Gaussian Distributions}
\begin{theorem}\label{thm:gauss}
    Let $\uip \overset{\mathrm{iid}}{\sim} \cN(\mu^*, I_d)$. 
    Let $\{z_i\}_{i=1}^n$ be an $\varepsilon$-corruption of $\uip$ where $C \leqslant \varepsilon < \frac{1}{2}$ for $C > 0$ sufficiently large. 
    Let $n = \Omega\left( d^{O(\log \frac{1}{\delta})}\right)$ where $\delta = 1 - 2\varepsilon$.  
    Then there is an algorithm based on SoS, running in time $n^{O\left( \log \frac{1}{\delta} \right)}$ that takes as input $\ip, \varepsilon$, and returns a vector $\hat{\mu}$ such that 
    \begin{align*}
        \Vert \hat{\mu} - \mus \Vert_2 = O\left( \sqrt{\log \frac{1}{\delta}} + \sqrt{\frac{d}{n}}\right)
    \end{align*}
    with probability 0.99.
\end{theorem}

We remark here that the $O\left(\sqrt{d/n}\right)$ term is necessary even without corruption. The remaining term of $O\left( \sqrt{\log \frac{1}{\delta}} \right)$ achieves the optimal dependence on $\varepsilon$ as $\varepsilon \to \frac{1}{2}$. To the best of our knowledge there is only an \emph{inefficient} estimator \cite{zhu2020does} that achieves this error, under a slightly weaker contamination model. This is thus the first that the optimal error of $O\left(\sqrt{\log \frac{1}{\delta}}\right)$ could be attained for this problem in time $n^{O\left( \log \frac{1}{\delta}\right)}$. We require $n = \Omega\left( d^{O(\log \frac{1}{\delta})}\right)$ for ensuring concentration of higher moments. We refer the reader to Section \ref{sec:gaussian} for a more detailed discussion on this problem.
\begin{lemma}[Lower Bound for Gaussians]
    There exists two distributions $\D_1$ and $\D_2$ on $\R$ such that for $C \leqslant \varepsilon < \frac{1}{2}$ for a sufficiently large constant $C > 0$ such that
    \begin{enumerate}
        \item $d_{TV}(\D_1, \D_2) \leqslant 2\varepsilon$.
        \item $\D_1$ and $\D_2$ are Gaussians with unit variance.
        \item $\left| \E_{X \sim \D_1}[X] - \E_{X \sim \D_2}[X] \right| \geqslant \Omega \left(\sqrt{\log \frac{1}{1 - 2\varepsilon}}\right)$.
    \end{enumerate}
\end{lemma}

Our results complement the results of existing robust mean estimation algorithms based on SoS \cite{kothari2018robust, hopkins2018mixture} -- We show that SoS based algorithms work for the full range of corruption and also attain \emph{optimal error} rates (up to constants) even as the corruption rate $\varepsilon \to \frac{1}{2}$. We emphasize here that there are few efficient algorithms that attain optimal breakdown point for this problem, and efficient algorithms achieving optimal error as $\varepsilon \to \frac{1}{2}$ were \emph{unknown} for different distributional assumptions. Our results make significant progress in this regard.

%% file: techniques.tex
\subsection{Technical Overview}\label{sec:techniques}

At a high level, SoS methods for robust mean estimation \cite{kothari2018robust} capture the following intuition using a SoS proof: Two distributions that are close in statistical distance, while belonging to a sufficiently nice distribution class must also have close means. They proceed by first setting up a system of polynomial constraints that captures the important aspects of the underlying distribution. This system additionally captures the fact that the input is an $\varepsilon$-corruption by encoding a constraint that asks the program variables to agree with the input on $(1-\varepsilon)$ fraction of the samples. A robust identifiability proof is then derived in SoS that essentially argues that any feasible solution of the polynomial system will be close to the true mean. This suffices as the goal becomes finding an appropriate degree SoS relaxation for (approximately) solving the polynomial system under consideration. We will stick to this blueprint when setting up the polynomial system in this work. 

When the corruption is large and close to $\frac{1}{2}$, the important thing as it turns out will be to focus on the overlap $\delta$ between the distributions and not the statistical distance. Indeed, the statistical distance and the overlap are complementary of one another: the sum of the statistical distance and the overlap between two distributions is 1\footnote{We refer the reader to Appendix \ref{app:factsoltv} for a discussion.}. We recall here that the way the canonical SoS program \cite{kothari2018robust} is set up, the statistical distance between the uniform distribution over the program variables and the uncorrupted samples is at most $2\varepsilon$, resulting in a corresponding overlap $\delta$ of at least $1 - 2\varepsilon$\footnote{Observe that $\varepsilon = \frac{1}{2}$ already implies that the distributions (and respectively their means) could be arbitrarily far apart. Note that this already indicates that SoS based methods have optimal breakdown point.}.  As long as the distributions belong to a sufficiently nice family, there is an intuitive relationship between the overlap of distributions and the distance between their means: the larger the overlap, the closer the means; the smaller the overlap, the farther apart the means. 

We obtain sharp error rates by capturing this intuition with a general identifiability proof when $\varepsilon \to \frac{1}{2}$ by exploiting the underlying properties of the distributions.  We then provide an SoS version of this proof that allows us to go from identifiability proofs to efficient algorithms. 

Indeed there is the question of why we need to view the problem from the perspective of overlap instead of statistical distance. The reason lies in how previous SoS proofs worked for our problem.

In the SoS proof introduced in \cite{kothari2018robust}, scaled and centered projections of data points that did not belong to the overlap region between the two distributions played a crucial role in bounding the distance between the means. The authors used appropriate SoS inequalities (SoS version of \Hoelder or Cauchy-Schwarz) to decouple the scaling and the centered projections. While the latter was controlled by the properties of the underlying distribution, the former related to the points not in the overlap and was a function of the statistical distance between the uniform distributions over the program variables and the uncorrupted samples\footnote{Recall that the statistical distance is at most $2\varepsilon$.}. This quantity provided the optimal dependency on $\varepsilon$ when $\varepsilon$ was sufficiently small and bounded away from $\frac{1}{2}$. When $\varepsilon \to \frac{1}{2}$, this quantity grows but is always bounded from above by a constant. Although we can obtain error that diverges as $\varepsilon \to \frac{1}{2}$ by being careful with the above approach and by introducing appropriate modifications (Theorems \ref{thm:optbdbcov}, \ref{thm:optbdmom}), decoupling the scaling and the centered projections the way in which it is done in current SoS proofs proves insufficient to attain \emph{sharp} error rates when $\varepsilon \to \frac{1}{2}$ for many distribution classes.

This step from optimal breakdown point to optimal error is what will thus require a new identifiability proof -- one that works primarily with overlap instead of statistical distance. In this proof we argue that the distance between the means is bounded by the sum of the distances between the individual means to the mean of the points in the overlap region between the distributions. As we show in our general identifiability proof for large corruption rates (Lemma \ref{lemma:generaliden}), bounding these distances still requires a decoupling step (into a scaling and a centered projection term) similar to the above case. However the key difference is that the scaling term is now a function of overlap, and this decoupling obtains the right error as $\varepsilon \to \frac{1}{2}$. For the bounded covariance case, a careful and optimized analysis inspired by existing identifiability proofs suffices to achieve optimal error (up to constants) for all $\varepsilon \in [0,\frac{1}{2})$. However, this happens to be a special case and these techniques do not generalize for other distribution classes. Specifically, as $\varepsilon \to \frac{1}{2}$ the general identifiability proof for large $\varepsilon$ recovers the optimal error for bounded covariance as a special case.

In summary, we study existing SoS proofs for robust mean estimation and through a careful analysis demonstrate their optimality in terms of the breakdown point. We then present a new proof of identifiability for robust mean estimation for large corruption rates. By designing an SoS version of this proof of identifiability, we achieve optimal error up to constants even as $\varepsilon \to \frac{1}{2}$ efficiently for a large class of distributions. 

%% file: notation.tex
\subsection{Notation}

We use the following convention: $\N$ is the set of natural numbers and $\R$ is the set of real numbers. $\R^d$ is the set of real vectors in $d$ dimensions.  For a positive integer $n$, $[n]$ is the set $\{1, 2, \dots, n\}$. Unless explicitly stated, the base of the logarithm is $e$. Unless otherwise specified, all vector norms $(\Vert . \Vert$ or $\Vert . \Vert_2)$ are the euclidean norm, and all matrix norms are the spectral norm. We use the notation $O(.), \Theta(.), \Omega(.), \lesssim, \gtrsim $ to hide absolute constants. We use $\mathbb{1}[.]$ for the indicator variable. We use $\widetilde{O}(.)$ to hide logarithmic factors. We denote the identity matrix in $d$ dimensions by $I_d$. Let $A, B \in \R^{d \times d}$. Then $A \preccurlyeq B$ is the ordering of $A$ and $B$ in \Lowner order. We use statistical distance interchangeably with Total Variation distance and denote it by $d_{TV}$. We use $\kldiv{X}{Y}$ to denote the Kullback-Leibler divergence between two probability distributions $X$ and $Y$. We use the notation $\uip$ and $\ip$ to denote uncorrupted and corrupted samples respectively. We use RHS and LHS to refer to the right and the left of an inequality respectively.

\subsection{Organization}
In Section \ref{sec:optbdcov}, we provide an optimized analysis that achieves optimal error up to constant factors for the bounded covariance case and discuss the challenges in extending this analysis for other cases. We overcome these challenges by presenting a new proof of identifiability for robust mean estimation when $\varepsilon$ is large in Section \ref{sec:iden} and subsequently capture this identifiability proof using a low-degree SoS proof to obtain efficient algorithms. We discuss the Gaussian case in Section \ref{sec:gaussian} before discussing the rounding procedure in Section \ref{sec:rounding}. We demonstrate the applicability of our identifiability proof to the problem of robust sparse mean estimation in Section \ref{sec:sparse}. In Appendix \ref{sec:lower} we provide information-theoretic lower bounds for the different distribution classes that we consider. We provide an introduction to the SoS proofs to algorithms framework in Appendix \ref{app:sos} and provide an SoS toolkit for all the SoS proofs that we use in Appendix \ref{app:sostool}.

%% file: algorithms.tex
\section{Robust Mean Estimation via Sum-of-Squares}\label{sec:main}
Throughout this paper, we always take sufficiently many samples to ensure that the distributional assumptions that we make also apply to the uniform distribution over the uncorrupted samples. The number of samples required to preserve these distributional assumptions also ensures that the sample mean of the uncorrupted samples is sufficiently close to the true mean. Consequently, the algorithms that we present in this paper will estimate the sample mean of the uncorrupted samples. Finally, by using triangle inequality, we obtain the desired estimate of the true mean.

\subsection{Optimal Error for Bounded Covariance}\label{sec:optbdcov}
We now consider the canonical SoS program introduced in \cite{kothari2018robust} for robust mean estimation. We provide an optimized analysis that achieves information-theoretically optimal error for all $\varepsilon \in [0,\frac{1}{2})$ up to constant factors for the bounded covariance case.  Consider the polynomial system \ref{eqn:basicsystem} in variables $\pwi$ and $\pxi$ where $w_i \in \R$ and $x_i \in \R^d$. $\pwi$ are indicators for if our program ``thinks'' a specific data point is uncorrupted or not. Observe that whenever $w_i$ is non-zero, we ask the program to agree with the corresponding data point. We capture using the second and third constraint that our input is an $\varepsilon$-corruption. The last two constraints enforce the fact that the empirical covariance of the program variables is bounded\footnote{ We observe that the positive semi-definiteness constraint can be implemented as a polynomial constraint by considering an \emph{auxiliary} matrix $R \in \R^{d \times d}$ such that $\Sigma + R^TR = I_d$. $R$ is auxiliary as we will not seek an assignment for it.}.

We note that this system is feasible for the following choice of  $\pwi$ and $\pxi$: $ \forall i \in [n]: w_i = \one\{z_i = x_i^*\}$ and $ \forall i \in [n]: x_i = x_i^*$. 
\begin{align}\label{eqn:basicsystem}
    \cA_{w, x} \coloneqq \left\{
    \begin{array}{l}
        \forall i \in [n]: \ w_i^2 = w_i   \\
        \forall i \in [n]: \ w_i\cdot (z_i - x_i) = 0\\
        \sn w_i \geqslant (1 - \varepsilon)\\
        \mu = \sn x_i \\
          \Sigma \coloneqq \sn (x_i - \mu)(x_i - \mu)^T \preccurlyeq I_d
        \end{array}
    \right\}
\end{align}
We now define the following constants that we will require in our proofs.
\begin{align}\label{eqn:truesamples}
    \forall i \in [n]: w_i^* \coloneqq \mathbb{1}\{z_i = x_i^*\}
\end{align}
Notice that these are precisely the ground truth indicators for if sample $x_i^*$ was uncorrupted.
Using this definition, we have the following Lemma that we will require in our proofs.
\begin{lemma}\label{lemma:thm}
\begin{align}
    \cA_{w,x} \sststile{2}{w} \left\{
    \begin{array}{l}
     \forall i \in [n]: w_iw_i^* \cdot x_i = w_iw_i^* \cdot x_i^*\\
     \forall i \in [n]: (1 - w_iw_i^*)^2 = (1 - w_iw_i^*) \\
     \sn (1 - w_iw_i^*) \leqslant 2\varepsilon\\
      1 - 2\varepsilon \leqslant \sn w_iw_i^*
    \end{array}
    \right\}
\end{align}
   
\end{lemma}
We defer the proof of the above \ref{lemma:thm} to Appendix \ref{app:proofs}. We are now ready to prove Theorem \ref{thm:boundedcov}. We recall the theorem statement before proving it.
\begingroup
\allowdisplaybreaks
\begin{theorem}
     Let $\D$ be a distribution on $\R^d$ with mean $\mu^*$ and covariance $\Sigma^*$. 
    Suppose $\Sigma^* \preccurlyeq \sigma^2\cdot I_d$. 
    Let $\uip \overset{\mathrm{iid}}{\sim} \D.$  
    Let $\{z_i\}_{i=1}^n$ be an $\varepsilon$-corruption of $\uip$ where $\varepsilon \in [0, \frac{1}{2})$.
    Let $n = \Omega\left( d \log d \right)$. 
    Then there is an efficient algorithm based on SoS, running in time $n^{O(1)}$ that takes as input $\ip, \varepsilon, \sigma$, and returns a vector $\hat{\mu}$ such that 
    \begin{align*}
        \Vert \hat{\mu} - \mus \Vert_2 = O\left( \sigma \cdot \sqrt{\frac{\varepsilon}{1-2\varepsilon}} + \sigma \cdot \sqrt{\frac{d}{n}}\right)
    \end{align*}
    with probability 0.99.
\end{theorem}
  
\begin{proof}
Overloading notation, we will use $\mu^* = \sn x_i^*$. As discussed earlier, it will suffice to estimate the sample mean of $\uip$. Moreover, using our assumption on $n$ in the above theorem we further assume that empirical covariance over the uncorrupted points is bounded from above by $I_d$ in \Lowner order for the purpose of our proof.  We refer the reader to related discussion in Lemma \ref{lemma:boundedcovconc}. We assume $\sigma = 1$ without loss of generality.
We now prove the following:
    \begin{align*}
        \cA_{w,x}  \sststile{O(1)}{w,x} \err^2 \leqslant O\left(\frac{\varepsilon}{1-2\varepsilon}\right).
    \end{align*}
This will suffice as the estimator we obtain after rounding will satisfy the same error guarantees as we show in Section \ref{sec:rounding}. Starting with the following degree 4 polynomial in $\mu$, we have
    \begin{align*}
        \err^4 &= \left(\err^2\right)^2 = \langle \del, \del \rangle^2 = \left(\sn \langle x_i - x_i^*, \del \rangle \right)^2 \\
        &= \left(\sn (1 - w_iw_i^*) \langle x_i - x_i^*, \mu - \mu^* \rangle \right)^2 \leqslant \left(\sn(1 - w_iw_i^*)^2\right)\left( \sn \langle x_i - x_i^*, \mu - \mu^* \rangle^2 \right) \\
        & = \left(\sn(1 - w_iw_i^*)\right)\left( \sn \langle x_i - x_i^*, \mu - \mu^* \rangle^2 \right) \leqslant 2\varepsilon \cdot \left( \sn \langle x_i - x_i^*, \mu - \mu^* \rangle^2 \right).
    \end{align*}
    The fourth equality is due to Lemma \ref{lemma:thm}. We used SoS Cauchy-Schwarz for the first inequality, then used Lemma \ref{lemma:thm} in the penultimate step. The last inequality follows from Lemma \ref{lemma:thm} as well. We now have:
\begin{align*}
    \sn \langle x_i - x_i^*, \mu - \mu^* \rangle^2 = \sn \langle x_i - \mu - x_i^* + \mu^* + \mu - \mu^*, \mu - \mu^* \rangle^2.
\end{align*}
Instead of splitting it into three terms and upper bounding them by SoS triangle inequality as in previous SoS works \cite{kothari2018robust, hopkins2018mixture}, we will instead simply expand the terms. We get
\begin{align*}
    &\sn \left( \langle x_i - \mu, \del \rangle - \langle x_i^* - \mu^*, \del \rangle + \err^2 \right)^2 \\
    &= \sn  \langle x_i - \mu, \del \rangle^2 + \sn \langle x_i^* - \mu^*, \del \rangle^2 + \sn \err^4   \\
    & \ \  -2 \left(\sn \langle x_i - \mu, \del \rangle \cdot \langle x_i^* - \mu^*, \del \rangle\right) - 2\left(\sn\langle x_i^* - \mu^*, \del \rangle \cdot \err^2 \right) \\
    & \ \ + 2 \left(\sn\langle x_i - \mu, \mu - \mu^* \rangle \cdot \err^2 \right). \\ 
\end{align*}
The first two squared terms can be bounded by the assumption on the covariance. Now we handle each cross term individually. We note that
\begin{align*}
    \left(\sn\langle x_i^* - \mu^*, \del \rangle \err^2 \right) &= \left( \left\langle \sn (x_i^* - \mu^*), \del \right\rangle \err^2 \right) \\ 
    &= \langle \mu^* - \mu^*, \del \rangle \err^2 = 0.
\end{align*}
Similarly,
\begin{align*}
    \left(\sn\langle x_i - \mu, \mu - \mu^* \rangle \err^2 \right) &= \left( \left\langle \sn (x_i - \mu), \del \right\rangle \err^2 \right) \\
    &=  \langle \mu - \mu, \del \rangle  \err^2 = 0.
\end{align*}
Now we have,
\begin{align*}
    -\left(\sn 2 \cdot \langle x_i - \mu, \del \rangle \cdot \langle x_i^* - \mu^*, \del \rangle\right) &\leqslant \sn \langle x_i - \mu, \del \rangle^2 + \sn \langle x_i^* - \mu^*, \del \rangle^2
\end{align*}
which follows from the fact that $(a+b)^2 \geqslant 0 \ \forall  a, b \in \R$.
Putting all the pieces together we get
\begin{align*}
    \sn \langle x_i - x_i^*, \mu - \mu^* \rangle^2 &\leqslant 2 \left(\sn \langle x_i - \mu, \del \rangle^2\right) + 2\left(\sn \langle x_i^* - \mu^*, \del \rangle^2 \right) + \err^4.
\end{align*}
Therefore we have
\begin{align*}
    \err^4 &\leqslant 2\varepsilon \cdot \left[ 2 \left(\sn \langle x_i - \mu, \del \rangle^2\right) + 2\left(\sn \langle x_i^* - \mu^*, \del \rangle^2 \right) \right]  + 2\varepsilon \cdot \err^4  \\
    &\leqslant 2\varepsilon \cdot \left( 2 \err^2 + 2 \err^2 \right) + 2\varepsilon \cdot \err^4.
\end{align*}
where the last inequality is due to the bounded covariance assumption. In particular we used our assumption on $n$ here for the uncorrupted samples. (Also see Lemma \ref{lemma:boundedcovconc}). Continuing with the proof, we get after rearranging that
\begin{align*}
    (1 - 2\varepsilon) \cdot \err^4 \leqslant 8\varepsilon \cdot \err^2.
\end{align*}
By SoS Cancellation,
\begin{align*}
    \err^2 \leqslant \frac{8\varepsilon}{1 - 2\varepsilon}.
\end{align*}
\end{proof}
\endgroup
 We emphasize here that previous works based on SoS focus on the regime when $\varepsilon$ is sufficiently small and bounded away from $\frac{1}{2}$, and that their proofs are inherently not optimized when $\varepsilon$ is large and close to $\frac{1}{2}$. We remark here that the proof for Theorem \ref{thm:boundedcov} is much more careful and optimized compared to standard SoS proofs \cite{kothari2018robust}. Our analysis also optimizes constants by avoiding inequalities such as SoS triangle inequality. Our proof for Theorem \ref{thm:boundedcov} starts with a degree $4$ polynomial $\left(\err^4\right)$ as opposed to a degree $2$ polynomial $\left(\err^2\right)$. We emphasize here that starting with a degree 2 polynomial suffices for both optimal error when $\varepsilon$ is small as well as optimal breakdown point (Theorem \ref{thm:optbdbcov}). The reason we start with a degree 4 polynomial is that when $\varepsilon$ is large and close to $\frac{1}{2}$, it allows us to bound $\delta \cdot \err^2$, where $\delta \coloneqq 1 - 2\varepsilon$ by an appropriate constant, allowing us to optimal dependence on $\delta$ as $\varepsilon \to \frac{1}{2}$.

Let us now consider the setting where the underlying distribution has certifiably bounded fourth moments.  Recall that distributions with certifiably bounded fourth moments also have bounded covariance. Therefore, Theorem \ref{thm:boundedcov} naturally applies to this case as well. However, information-theoretically it is possible to exploit this stronger assumption to get better dependency on both $\varepsilon$ and $\delta$ when $\varepsilon$ is small and large respectively. In particular, for bounded fourth moment distributions, the optimal dependency on $\varepsilon$ and $\delta$ for small enough $\varepsilon$ and $\delta$ are $O(\varepsilon^{3/4})$ and $O(\delta^{-1/4})$ respectively. 

If we want to build on the proof in Theorem \ref{thm:boundedcov}, a possible strategy would be to start with $\err^8$, and then replace SoS Cauchy-Schwarz, with SoS \Hoelder. We remark here that this is the decoupling step that we discussed in Section \ref{sec:techniques}. This will result in a bound of $8 \cdot \varepsilon^3$ for the first term. While this constant factor is already sub-optimal (notice that we require $2\varepsilon \cdot \err^8$ on the RHS for even optimal breakdown point), the second term once expanded will have cross moment terms that do not cancel out due to asymmetry unlike the bounded covariance case. While we can indeed bound the second term with inequalities like SoS triangle inequality or related variants as has been done in \cite{kothari2018robust, hopkins2018mixture} these lead to sub-optimal dependency on $\varepsilon$ and the analysis will indicate a much worse breakdown point instead of the correct one of $\frac{1}{2}$. Even in the proof of Theorem \ref{thm:boundedcov} above, using SoS triangle inequality will prove sub-optimal. These limitations in existing SoS proofs, especially in the difficulty in obtaining the right dependency on $\varepsilon$ after the decoupling step motivates our identifiability proof for large corruption rates.

\subsection{Identifiability for large corruption}\label{sec:iden}
In previous identifiability proofs for robust mean estimation, the focus was on bounding the distance between the means when the statistical distance between two distributions from a nice distribution class was small. The SoS version presented above was inspired by the initial identifiability proof of \cite{kothari2018robust} which focused on the regime where $\varepsilon$ was small. When $\varepsilon$ is large the means are inherently far away since the statistical distance is large. The distance between the means when $\varepsilon$ is large is captured more sharply by the overlap $\delta$ between the distributions, rather than the statistical distance as discussed earlier. We will utilize this notion of overlap to formalize the following intuition in the next Lemma: If two distributions from a sufficiently nice class have large overlap, then their means must be close and if they have small overlap, their means must be far.
\begingroup
\allowdisplaybreaks
\begin{lemma}[Identifiability when $\varepsilon$ is close to $\frac{1}{2}$]\label{lemma:generaliden}
    Let $\D_1, \D_2$ be distributions on $\R^d$ with means $\mu_1$ and $\mu_2$ respectively. Assume that $\D_1$ and $\D_2$ have bounded $k^{th}$ moments in the sense of Definition \ref{defn:certifiablebounded}.  Suppose $d_{TV}(\D_1, \D_2) \leqslant 1 - \delta$ for $\delta$ sufficiently small. Then,
    \begin{align*}
        \vnorm{\mu_1 - \mu_2} \leqslant O\left(\frac{\sqrt{k}}{\delta^{1/k}}\right)
    \end{align*}
\end{lemma}
\begin{proof}
    Consider a coupling $\D_1(x)$ and $\D_2(y)$ such that $\Pr[x = y] \geqslant \delta$. Then, for any fixed vector $v \in \R^d$
    \begin{align*}
        \left|\langle \mu_1 - \mu_2, v \rangle \right| = \left|\langle \mu_1 - \mu' + \mu' - \mu_2, v \rangle \right| \leqslant \left| \langle \mu_1 - \mu', v \rangle\right| + \left| \langle \mu' - \mu_2, v \rangle\right| 
    \end{align*}
    where $\mu' \coloneqq \Ex{X|A}$ and $A$ is the event $x = y$ with $x \sim \D_1$ and $y \sim \D_2$.
     The above proof has the following interpretation: along any direction $v$, the distance between the means $\mu_1 - \mu_2$ is at most sum of the distances between the means ($\mu_1, \mu_2)$ to the mean of the common region of \emph{overlap} ($\mu'$).
    We now analyze the first term on the right hand side.
    \begin{align*}
        \left| \langle \mu_1 - \mu', v \rangle\right| &= \left| \langle \mu_1 - \E_{X \sim \D_1}[X|A], v \rangle\right| =\left| \E_{X \sim \D_1}\left[\langle X - \mu_1, v \rangle | A \right]\right| 
        = \frac{\left| \E_{X \sim \D_1} \left[ \one[A] \cdot \langle X - \mu_1, v \rangle \right] \right|}{\Pr[A]}. 
    \end{align*}
     For the above equalities, we used standard facts from conditional probability. Notice that by this point, we can already decouple the distance into a scaling and a moment term. We now have
     \begin{align*}
        \frac{\left| \E_{X \sim \D_1} \left[ \one[A] \cdot \langle X - \mu_1, v \rangle \right] \right|}{\Pr[A]}   &\leqslant \frac{\E_{X \sim \D_1}\left[\one[A]^{k/(k-1)}\right]^{1 - 1/k}\cdot \E_{X \sim \D_1}\left[ \langle X - \mu_1, v \rangle^k \right]^{1/k}}{\Pr[A]} \\
        & \leqslant \frac{\Pr[A]^{1-1/k} \cdot \left(k^{k/2} \Vert v \Vert^{k}\right)^{1/k}}{\Pr[A]} \\
        &= \frac{\sqrt{k} \Vert v \Vert}{\Pr[A]^{1/k}} \leqslant \frac{\sqrt{k}\Vert v \Vert}{\delta^{1/k}}
     \end{align*}
     where the first inequality follows from \Hoelder's inequality and the second inequality from the assumptions on the moments. 
    Similarly, we have
    \begin{align*}
        \left| \langle \mu' - \mu_2, v \rangle\right| \leqslant \frac{\sqrt{k}\Vert v \Vert}{\delta^{1/k}}.
    \end{align*}
    Putting the pieces together and picking $v = \mu_1 - \mu_2$, we get that
    \begin{align*}
        \vnorm{\mu_1 - \mu_2} \leqslant O\left(\frac{\sqrt{k}}{\delta^{1/k}}\right).
    \end{align*}
\end{proof}
\endgroup
Observe that for the bounded covariance ($k=2$) case, the above identifiability proof already recovers the optimal error of Theorem \ref{thm:boundedcov}. We remark here that the above Lemma is not new in the literature. Lemma \ref{lemma:generaliden} and a related SoS version of it (See Lemma \ref{lemma:ineff}) already appear in \cite{hopkins2018clustering, hopkins2018mixture} in the context of \emph{clustering}. This is the first that this identifiability proof has instead been utilized to obtain results for the problem of robust mean estimation, especially for large corruption rates. In retrospect, the reason is that for clustering we require small overlap between clusters. For robust mean estimation near the breakdown point, the means are far apart and the two distributions of interest have small overlap.  We remark here that the SoS version of the proof used in \cite{hopkins2018clustering, hopkins2018mixture} is insufficient for our purpose as it only infers statements about $\err^{2k}$. This sufficed for their goal of showing that if there existed a subset of points with bounded moments, then this subset had large intersection with at least one of the clusters (Lemma 1 in \cite{hopkins2018clustering}). The goal in our problem is instead to show in SoS that the distance of the means between two bounded moment distributions can be sharply controlled by their overlap. Their SoS inequality alone does not suffice for our goal as it cannot obtain optimal dependency on the overlap. We instead provide a simple and (stronger) SoS proof that achieves the desired error.

\subsection{SoSizing the proof of identifiability}
We now consider the following polynomial system in variables $\pxi \in \R^d$ and $\pwi \in \R$. Consistent with before, let $\ip$ be an $\varepsilon$-corruption of $\uip$. 

\begin{align}\label{eqn:finalsystem}
     \cA \coloneqq \left\{
    \begin{array}{l}
        \forall i \in [n]: \ w_i^2 = w_i   \\
        \forall i \in [n]: \ w_i\cdot (z_i - x_i) = 0\\
        \sn w_i \geqslant (1 - \varepsilon) \\
        \mu = \sn x_i \\
        \forall v \in \R^d : \sn \langle x_i - \mu, v \rangle^k \leqslant k^{k/2}\Vert v \Vert^k 
        \end{array}
    \right\}
\end{align}
The above polynomial system $\cA$ is essentially the same as $\cA_{w,x}$ (System \ref{eqn:basicsystem}), except we replace the bounded covariance constraint with the appropriate moment condition. We note that feasibility is straightforward to check for this system, similar to system $\cA_{w,x}$. We also note that the final constraint above is a \emph{universal} constraint. We refer the reader to the Appendix \ref{app:facts} to see how to model this inside SoS. We recall Theorem \ref{theorem:general} before providing the formal proof.

\begin{theorem}
Let $\D$ be a distribution on $\R^d$  with mean $\mu^*$ and certifiably bounded $k^{th}$ moments in the sense of Definition \ref{defn:certifiablebounded} where $k$ is a power of 2.  
Let $\uip \overset{\mathrm{iid}}{\sim} \D$. Let $\{z_i\}_{i=1}^n$ be an $\varepsilon$-corruption of $\uip$ where $\varepsilon \in [0, \frac{1}{2})$.
Let $n = \Omega\left(d^{O(k)}\right)$.  
Then there is an efficient algorithm based on SoS, running in time $n^{O(k)}$ that takes as input $\ip, \varepsilon, k$, and returns a vector $\hat{\mu}$ such that 
    \begin{align*}
        \Vert \hat{\mu} - \mus \Vert_2 = O\left( \frac{\sqrt{k}\cdot \varepsilon^{1-1/k}}{(1-2\varepsilon)^{1/k}} + \sqrt{\frac{d}{n}}\right)
    \end{align*}
        with probability 0.99.
\end{theorem}
We will prove the result for the case when $\varepsilon \to \frac{1}{2}$ inspired by the identifiability proof in Lemma \ref{lemma:generaliden}. As mentioned earlier, the result for small $\varepsilon$ is well-known \cite{kothari2018robust, hopkins2018mixture}\footnote{We also refer the reader to our analysis for optimal breakdown point (Theorem \ref{thm:optbdmom}), which recovers this result.}. Putting together the results for both the regimes of corruption gives us the guarantee of Theorem \ref{theorem:general}.
\begin{proof}
 We overload notation like we did in the proof of Theorem \ref{thm:boundedcov} -- we will use $\mu^* = \sn x_i^*$. As argued before, estimating the sample mean of $\uip$ is enough. Defining $\delta \coloneqq 1 - 2\varepsilon$,
     we will prove the following:
    \begin{align*}
        \cA \sststile{O(k)}{w,x,v} \err^2  \leqslant O\left( \frac{k}{\delta^{2/k}}\right).
    \end{align*}
As with Theorem \ref{thm:boundedcov}, this is sufficient as the estimator we obtain after rounding will satisfy the same error guarantees. We have
\begin{align*}
    \delta \cdot \langle \del , v \rangle^k \leqslant \left( \sn w_iw_i^* \right) \cdot \langle \del, v \rangle^k.
\end{align*}
Here the inequality follows from the definition of $\delta$ and Lemma \ref{lemma:thm}. Then,
\begin{align*}
    \left( \sn w_iw_i^* \right) \cdot \langle \del, v \rangle^k &= \left( \sn \left(w_iw_i^* \right)^k \right) \cdot \langle \del, v \rangle^k \\ 
    &= \sn \left(w_iw_i^* \cdot \langle \mu - \mu^*, v \rangle \right)^k \\
    &= \sn \left(w_iw_i^* \cdot \langle \del, v \rangle + w_iw_i^* \cdot \langle x_i^* - x_i, v \rangle \right)^k.
\end{align*}
Here the first equality is due to the booleanity of $w_iw_i^*$, and the final equality is because of Lemma \ref{lemma:thm}. Rearranging the terms, we obtain
\begin{align*}
    \delta \cdot \langle \del , v \rangle^k \leqslant \sn \left( w_iw_i^* \cdot \langle \mu - x_i, v \rangle + w_iw_i^* \cdot \langle x_i^* - \mu^*, v \rangle \right)^k.
\end{align*}
Applying SoS Triangle Inequality to the Right Hand Side, we get
\begin{align*}
    \delta \cdot \langle \del , v \rangle^k &\leqslant 2^{k-1}\left(\sn \left(w_iw_i^* \cdot \langle \mu - x_i, v \rangle \right)^k + \sn \left(w_iw_i^* \cdot \langle x_i^* - \mu^*, v \rangle \right)^k \right)\\
    &= 2^{k-1}\left(\sn (w_iw_i^*)^k \cdot \langle \mu - x_i, v \rangle^k + \sn (w_iw_i^*)^k \cdot \langle x_i^* - \mu^*, v \rangle^k \right) \\
    &= 2^{k-1}\left(\sn w_iw_i^* \cdot \langle \mu - x_i, v \rangle^k + \sn w_iw_i^* \cdot \langle x_i^* - \mu^*, v \rangle^k \right) \\
    &\leqslant 2^{k-1}\left(\sn \langle \mu - x_i, v \rangle^k + \sn  \langle x_i^* - \mu^*, v \rangle^k \right) \\
    &\leqslant 2^{k-1} \left(2 \cdot k^{k/2} \cdot  \Vert v \Vert^k \right) = 2^k \cdot k^{k/2} \cdot \Vert v \Vert^k.
\end{align*}
where the second inequality is because $\forall i \in [n] : w_iw_i^* \leqslant 1 $ (See Proof of Lemma \ref{lemma:thm} in Appendix \ref{app:lemmaproof}) and the last inequality follows from our assumption of certifiably bounded higher moments in the sense of Definition \ref{defn:certifiablebounded}. Putting all these pieces together, we get
\begin{align*}
    \delta \cdot \langle \del , v \rangle^k \leqslant 2^k \cdot k^{k/2} \cdot \Vert v \Vert^k.
\end{align*}
Now picking $v = \del$, we have
\begin{align*}
    \delta \cdot \err^{2k} \leqslant 2^k \cdot k^{k/2} \cdot \err^k.
\end{align*}
By SoS Cancellation,
\begin{align*}
    \err^k \leqslant \frac{2^k \cdot  k^{k/2}}{\delta}.
\end{align*}
By finally taking SoS Square Root, we get
\begin{align*}
   \cA \sststile{O(k)}{w,x,v}  \err^2  \leqslant \frac{4\cdot k}{\delta^{2/k}} = O \left( \frac{k}{\delta^{2/k}} \right).
\end{align*}
\end{proof}

We note that the above proof is much simpler compared to standard SoS proofs that exploit inequalities such as SoS \Hoelder when dealing with higher moments. 
We now present the SoS version of Lemma \ref{lemma:generaliden} that appears in \cite{hopkins2018clustering,hopkins2018mixture}.
\begin{lemma}[Lemma 5.5 of \cite{hopkins2018mixture}]\label{lemma:ineff}
    \begin{align*}
        \cA \sststile{O(k)}{w,x,v} \left(\sn w_iw_i^* \right)^{k} \cdot \err^{2k} \leqslant 2^k \cdot k^{k/2} \cdot \left(\sn w_iw_i^* \right)^{k-1} \cdot \err^k
    \end{align*}
\end{lemma}
    We defer the proof of Lemma \ref{lemma:ineff} to Appendix \ref{app:proofs}. We note that cancellation of terms is not always allowed in the SoS proof system. Indeed if we could cancel the terms above, we will get the desired error we obtained in Theorem \ref{theorem:general}. We discuss more about this in Appendix \ref{app:factscompare} and show how Lemma \ref{lemma:ineff} and Theorem \ref{theorem:general} are related. 

\subsection{Robust Gaussian Mean Estimation with known Covariance}\label{sec:gaussian}

We now consider the problem of robust mean estimation when input $\ip$ is an $\varepsilon$-corruption of samples  $\uip$ where $x_i^* \overset{\mathrm{iid}}{\sim} \cN(\mu^*, \Sigma)$. Here we assume $\varepsilon$ to be sufficiently large and also assume knowledge of the covariance $\Sigma$. Through a straightforward whitening procedure\footnote{We simply multiply the datapoints with the matrix $\Sigma^{-1/2}$.} the problem reduces to robustly estimating the mean of an identity covariance Gaussian.  

We note here that the class of distributions considered in Definition \ref{defn:certifiablebounded} includes the identity covariance Gaussian distribution as well \cite{hopkins2018mixture}. Therefore the results of Theorem \ref{theorem:general} apply to the Gaussian case. Specifically, by taking $k$ moments (where $k$ is a constant), there is an efficient algorithm that robustly estimates the mean of a Gaussian to error $O \left(\sqrt{k}\cdot (1 - 2 \varepsilon)^{-1/k} \right)$ when $\varepsilon$ is large. 

However, the optimal error for robust Gaussian mean estimation when $\varepsilon$ is large is $O\left(\sqrt{\log \frac{1}{\delta}}\right)$, where $\delta = 1 - 2\varepsilon$.  In order to achieve error $O\left(\sqrt{\log \frac{1}{\delta}}\right)$, Theorem \ref{theorem:general} requires $\log \frac{1}{\delta}$ many moments which translates to a sample complexity $n \gtrsim d^{O\left({\log \frac{1}{\delta}}\right)}$ for concentration of measure arguments and running time $n^{O \left(\log \frac{1}{\delta}\right)} = d^{O\left( \log \frac{1}{\delta} \right)^2}$. Our algorithm thus requires running time that is quasi-polynomial in $\delta^{-1}$ - i.e., quasi-polynomial in the inverse of the distance from the breakdown point. We therefore obtain the results of Theorem \ref{thm:gauss} as a Corollary of our more general Theorem \ref{theorem:general}.

\begin{corollary}
 Let $\uip \overset{\mathrm{iid}}{\sim} \cN(\mu^*, I_d)$. 
    Let $\{z_i\}_{i=1}^n$ be an $\varepsilon$-corruption of $\uip$ where $C \leqslant \varepsilon < \frac{1}{2}$ for $C > 0$ sufficiently large. 
    Let $n = \Omega\left( d^{O(\log \frac{1}{\delta})}\right)$ where $\delta = 1 - 2\varepsilon$.  
    Then there is an algorithm based on SoS, running in time $n^{O\left( \log \frac{1}{\delta} \right)}$ that takes as input $\ip, \varepsilon$, and returns a vector $\hat{\mu}$ such that 
    \begin{align*}
        \Vert \hat{\mu} - \mus \Vert_2 = O\left( \sqrt{\log \frac{1}{\delta}} + \sqrt{\frac{d}{n}}\right)
    \end{align*}
    with probability 0.99.    
\end{corollary}
This is the first that the optimal error of $O\left( \sqrt{\log \frac{1}{\delta}} \right)$ could be attained for this problem in time $n^{O\left( \log \frac{1}{\delta} \right)}$. To the best of our knowledge there is only an \emph{inefficient} estimator \cite{zhu2020does} that achieves this error, under the slightly weaker total variation contamination model\footnote{In the total variation contamination model, the distribution is first corrupted and then samples are generated from this corrupted distribution. Strong contamination model is a stronger model of corruption as the adversary has knowledge of the samples directly.}. The inefficiency in \cite{zhu2020does} is due to projection under the generalized Kolmogorov-Smirnov distance (See Appendix B of \cite{zhu2022generalized} for more discussion). \cite{zhu2020does} do not explicitly analyze their error guarantee as $\varepsilon \to \frac{1}{2}$. We analyze their estimator in  Theorem \ref{thm:zjs}, and demonstrate that it achieves this optimal error.  

\subsubsection{Discussion about existing work}
When $\varepsilon$ is small, the information-theoretically optimal error is $O(\varepsilon)$. Efficient algorithms achieve error of $\ot$\footnote{The exact term is $O\left(\varepsilon \sqrt{\log \frac{1}{\varepsilon}}\right)$} \cite{diakonikolas2019robust} and there is evidence that the additional logarithmic factor is necessary for efficient Statistical Query algorithms \cite{diakonikolas2017statistical} in the total variation contamination model\footnote{This is therefore also a lower bound against statistical query algorithms in the strong contamination model.}.

In order to achieve $\ot$ error, \cite{diakonikolas2019robust} utilize strong concentration properties of Gaussians up to the first two moments. The SoS based algorithm in \cite{kothari2018robust, hopkins2018mixture} attains this error in running time $n^{O(\log \frac{1}{\varepsilon})}$ requiring $O\left(\log \frac{1}{\varepsilon}\right)$ moments. This left open a question of if there was an inherent drawback with utilizing SoS for the Gaussian case, or if it was an artifact of the analysis.

\cite{kothari2022polynomial} demonstrated that with a more involved analysis that the near-optimal error of $\ot$ error is achievable in \emph{polynomial} time using SoS. Their key idea was to interleave proofs in SoS and in pseudo-expectation. This allowed them to utilize the same concentration properties of Gaussians used in \cite{diakonikolas2019robust} outside SoS, circumventing the need to certify these properties \emph{within} SoS.

The improved analysis from \cite{kothari2022polynomial} above however does not extend for the case when $\varepsilon$ is large. Specializing the techniques of \cite{kothari2022polynomial} to our setting, we can interpret the proof in \cite{kothari2022polynomial} as formalizing the following: A Gaussian and a specific bounded covariance distribution have means $\ot$ apart as long as they are sufficiently close in statistical distance.
We observe here that this argument is \emph{asymmetric} as opposed to the ones we have considered so far in this paper -- the two distributions belong to different families. When $\varepsilon$ is large, the tails of the individual distributions play a significant role.  The overlap region when $\varepsilon$ is large is essentially at the tail of the distributions. In particular, the distance between the means will be dominated by the distance from the mean of the bounded covariance distribution to the overlap region between the distributions. Therefore, we would only be able to infer the same error of $O(\delta^{-1/2})$ that we obtain in Theorem \ref{thm:boundedcov}. In particular, specializing the program of  \cite{kothari2022polynomial} for the identity covariance case alone would not be sufficient to improve over $O(\delta^{-1/2})$.
We refer the reader to Appendix \ref{app:gauss} for more details and lower bounds against algorithms that estimate the mean of a Gaussian using a bounded covariance distribution.

For the Gaussian case, \cite{dalalyan2022all} designed an efficient robust estimator that achieves optimal breakdown point. When $\varepsilon < (5 - \sqrt{5})/10 \approx 0.28$, they achieve near-optimal dimension independent error of $\ot$. However when $\varepsilon$ is large, their estimator is the geometric median, whose error scales with the dimension. In particular, when $\varepsilon$ is large, the geometric median has error $O(\sqrt{d}\cdot \delta^{-1})$ (Lemma 2 of \cite{dalalyan2022all}). \cite{diakonikolas2023algorithmic} suggest that a variant of the filter method achieves optimal breakdown point for the Gaussian case efficiently. The best error rate attained by filter based methods as $\varepsilon \to \frac{1}{2}$ is $O(\delta^{-1})$ \cite{zhu2022robust}, albeit for the bounded covariance case. We believe that the approach suggested in \cite{diakonikolas2023algorithmic} achieves similar error.

It remains an intriguing open question if information-theoretically optimal rate of $O\left(\sqrt{\log \frac{1}{\delta}}\right)$ for Gaussians can be achieved in polynomial time.

\subsection{Rounding}\label{sec:rounding}
Our rounding procedure is simple and is the same as the ones used in other SoS based algorithms for robust mean estimation \cite{kothari2018robust, hopkins2018mixture}. Before proving the guarantees we present the high-level algorithm for completeness.

\begin{algorithmbox}[Robust Mean Estimation via SoS]
  \label{algo:mean}
  \mbox{}\\
  \textbf{Input:} $\varepsilon$-corrupted samples $\ip$ from a specific distribution, $\varepsilon$, and other problem parameters.

  \noindent
  \textbf{Output:} $\hat{\mu} \coloneqq \pE[\mu]$ where $\pE$ is an appropriate degree pseudo-expectation satisfying a specific polynomial system.
\end{algorithmbox}

For the Polynomial System $\cA_{w,x}$ (System \ref{eqn:basicsystem}), find \emph{any} degree 6 pseudo-expectation $\pE$ satisfying $\cA_{w,x}$. Return the estimator $\hat{\mu} \coloneqq \pE[\mu]$. Similarly, for the Polynomial System $\cA$ (System \ref{eqn:finalsystem}), find \emph{any} degree $10k$ pseudo-expectation $\pE$ satisfying $\cA$ and return the estimator $\hat{\mu} \coloneqq \pE[\mu]$. Now we have, for both cases,
\begin{align*}
    \vnorm{\hat{\mu} - \mu^*}^2 &= \Vert \pE[\mu] - \mu^* \Vert^2 = \langle \pE[\mu] - \mu^*, \pE[\mu] - \mu^* \rangle \\
    &= \pE[\langle \mu - \mu^*, \pE[\mu] - \mu^* \rangle] \leqslant \frac{1}{2} \cdot \pE \left[ \err^2 + \Vert \pE[\mu] - \mu^* \Vert^2 \right]
\end{align*}
where we used the SoS inequality $(a-b)^2 \geqslant 0$ for each coordinate of the vector. Rearranging, we get,
\begin{align*}
    \Vert \pE[\mu] - \mu^* \Vert^2 \leqslant \pE\left[\err^2\right]
\end{align*}
We are done at this point as we have SoS proofs that bound $\err^2$ by the appropriate function of $\varepsilon$, and pseudo-expectations preserve SoS inequalities.

%% file: sparse.tex
\subsection{Application: Robust Sparse Mean Estimation}\label{sec:sparse}

In this section we demonstrate the applicability of our approach to the problem of robust sparse mean estimation. We describe the problem below and explain how we can apply Theorem \ref{theorem:general} in a black-box manner to obtain information-theoretically optimal error when $\varepsilon$ is large.

Given $\varepsilon$-corrupted samples from a distribution $\D$ on $\R^d$ with $k$-sparse mean $\mu^* \in \R^d$, where $k \ll d$, the goal in robust sparse mean estimation is to return $\hat{\mu}$ such that $\vnorm{\hat{\mu} - \mu^*}$ is small. We recall that the set of $k$-sparse vectors consists of all vectors in $\R^d$ with at most $k$ non-zero entries.
\begin{definition}[$(M,t)$ certifiably bounded moments in $k$-sparse directions. \cite{diakonikolas2022robust}]\label{defn:sparsebounded}
     For some $M > 0$ and even $t \in \mathbb{N}$, we say that distribution $\D$ with mean $\mu$ satisfies $(M,t)$ certifiably bounded moments in $k$-sparse directions if 
\begin{align*}
    \cA_k \sststile{O(t)}{v, z} \E_{X \sim \D}\left[\langle v, X - \mu \rangle^t\right]^2 \leqslant M^2
\end{align*}
where
\begin{align}\label{eqn:ksparse}
  \cA_k \coloneqq \left\{
    \begin{array}{l}
    \forall i \in [d]: z_i^2 = z_i \ \ \\
    \forall i \in [d]: v_iz_i = v_i \ \ \\
    \sum\limits_{i=1}^d z_i \leqslant k \ \ \\
    \sum\limits_{i=1}^d v_i^2 = 1
\end{array}\right\}
\end{align}
\end{definition}

A vector $v \in \R^d$ is $k$-sparse (with $\Vert v \Vert = 1$) if and only if $ \exists \ z \in \R^d$ such that $v$ and $z$ satisfy $\mathcal{A}_k$. \cite{diakonikolas2022robust} show that sampling preserves $(M, t)$ certifiably bounded moments in $k$-sparse directions for distributions $\D$ satisfying Definition \ref{defn:sparsebounded} and having sufficiently light tails. We will work in this setting, similar to the case of certifiably bounded moments. Our goal will therefore be to robustly estimate the $k$-sparse mean $\mu^*$ from $\varepsilon$-corrupted samples of an $(M,t)$ certifiably bounded distributions in $k$-sparse directions when $\varepsilon$ is large.

We remark here that the optimal error achievable in both $k$-sparse setting and the dense setting is the same for $\varepsilon \in (0, \frac{1}{2})$. In particular, for $\D$ satisfying Definition \ref{defn:sparsebounded}, the information-theoretically optimal dependence on $\varepsilon$ is $O(\varepsilon^{1 - 1/t})$ for $\varepsilon$ small and $O((1 - 2\varepsilon)^{-1/t})$ when $\varepsilon$ is large.  

The focus in the sparse setting is on sample complexity. Recall from the previous section that for concentration of measure arguments we require $\approx d^{O(t)}$ samples assuming certifiably bounded moments up to order $t$. In the sparse case however, the expectation on sample complexity is $\text{poly}(k, \log d)$ samples.
\cite{diakonikolas2022robust} build on the SoS program of \cite{kothari2018robust} for distributions satisfying Definition \ref{defn:sparsebounded} and obtain information-theoretically optimal error of $O\left(\varepsilon^{1-1/t}\right)$ with sample complexity dependence $\text{poly}(k, \log d)$ when $\varepsilon$ is sufficiently small. In particular their proof is not optimized and requires $\varepsilon < 0.003$. 

We demonstrate that the analysis in Theorem \ref{theorem:general} will go through by utilizing the moment upper bounds specific to this case, attaining information-theoretically optimal dependence on $\delta = 1 - 2\varepsilon$ of $O(\delta^{-1/t})$ for $\varepsilon$ close to $\frac{1}{2}$, which is the first for the problem of robust sparse mean estimation. We retain the sample complexity of $\text{poly}(k, \log d)$ since ours is an improved analysis of the program in \cite{diakonikolas2022robust} for large $\varepsilon$. 
We now specify a few useful definitions and facts that will make our goals clear. We recall the definition of the $(2,k)$ norm of a vector:
\begin{align*}
\Vert x \Vert_{2,k} \coloneqq \max_{\Vert v \Vert_2 = 1 \ v: k\text{-sparse}} \langle v, x \rangle .
\end{align*}
We also define the following set: $S_k = \{ v \in \R^d : \Vert v\Vert_2 = 1, v \text{ is }k\text{-sparse} \}$. 
\begin{fact}[Fact 9 from \cite{diakonikolas2022robust}]
 Let $h_k(x): \R^d \to \R^d$  be the function that truncates $x$ to its largest $k$ values and zeros out the rest of the values. For all $a \in \R^d$ that are $k$-sparse, we have that 
\begin{align*}
    \Vert h_k(x) - a \Vert \leqslant 3 \Vert x - a \Vert_{2,k}
\end{align*}   
\end{fact}

In particular, when  $x = \hat{\mu}$ (our estimator) and $a = \mu^*$ (the true mean of the distribution $\D$) we have estimation error,
\begin{align*}
    \Vert h_k\left(\hat{\mu}\right) - \mu^* \Vert \leqslant 3 \Vert \hat{\mu} - \mu^* \Vert_{2,k}.
\end{align*}
Therefore it is sufficient to bound the $(2,k)$ norm of $ \hat{\mu} - \mu^*$. In particular, it is sufficient to show that the our estimator $\pE[\mu]$ is such that $ \Vert \pE[\mu] - \mu^* \Vert_{2,k}$ is small. Indeed our aim will be to show that
\begin{align*}
   \langle v, \pE[\mu] - \mu^* \rangle = \pE \langle v, \mu - \mu^* \rangle.
\end{align*}
is small for all vectors $v$ that are $k$-sparse. To this end, we will show that
\begin{align*}
     \forall v \in S_k : \left| \pE \langle v, \mu - \mu^* \rangle \right| &\leqslant \left| \pE \langle v, \mu - \mu' \rangle \right| + \left| \pE \langle v, \mu' - \mu^* \rangle \right|\\
     &= \left| \pE \langle v, \mu - \mu' \rangle \right| + \left|\langle v, \mu' - \mu^* \rangle \right|.
\end{align*}
where $\mu'$ is the sample mean of the uncorrupted samples. By concentration of measure arguments (Lemma 14 of \cite{diakonikolas2022robust}) the second term will be small. It will thus suffice to bound the first term. We now define the polynomial system and present the formal statement.

Consider the following polynomial system in variables $\pxi \in \R^d$ and $\pwi \in \R$. Consistent with before, let $\ip$ be an $\varepsilon$-corruption of $\uip$. 

\begin{align}\label{eqn:sparsesystem}
     \cA_{\text{sparse}} \coloneqq \left\{
    \begin{array}{l}
        \forall i \in [n]: \ w_i^2 = w_i   \\
        \forall i \in [n]: \ w_i\cdot (z_i - x_i) = 0\\
        \sn w_i \geqslant (1 - \varepsilon) \\
        \mu = \sn x_i \\
        \forall v \in S_k:  \left[\sn \langle v, x_i - \mu \rangle^t\right]^2 \leqslant M^2
        \end{array}
    \right\}
\end{align}

\begin{theorem}\label{theorem:sparse}
Let $\D$ be a distribution on $\R^d$ with $k$-sparse mean $\mu^*$ and $(M,t)$ certifiably bounded moments in $k$-sparse directions in the sense of Definition \ref{defn:sparsebounded} where $t$ is a power of 2.  
Let $\uip \overset{\mathrm{iid}}{\sim} \D$.  Let $\{z_i\}_{i=1}^n$ be an $\varepsilon$-corruption of $\uip$ where $\varepsilon \in [0,\frac{1}{2})$.
Let $n = \Omega\left(\left( t k \left(\log d\right)\right)^{O(t)} \right)$. Then there is an efficient algorithm based on SoS, running in time $(n\cdot d)^{O(t)}$ that takes as input $\ip,k,M,\varepsilon, t$, and returns a vector $\hat{\mu}$ such that 
\begin{align*}
    \vnorm{\hat{\mu} - \mu^* }_{2,k} \leqslant O\left(  \frac{M^{1/t} \cdot \varepsilon^{1 - 1/t}}{(1 - 2\varepsilon)^{1/t}}\right)
\end{align*}
with probability 0.99.
\end{theorem}
We remark here that Theorem \ref{theorem:sparse} presents a complete picture for robust sparse mean estimation under the assumption of Definition \ref{defn:sparsebounded} for $\varepsilon \in (0,\frac{1}{2})$. As mentioned above, Lemma 14 of \cite{diakonikolas2022robust} demonstrates that for the number of samples we take in Theorem \ref{theorem:sparse}, the sample mean can be made sufficiently close to the true mean, and we therefore omit those details in our theorem statement. We have thus demonstrated a computationally efficient estimator based on SoS that attains the information-theoretically optimal error (up to) constants for the full range of corruption while maintaining sample complexity poly$(k, \log d)$.

We will prove the result for the case when $\varepsilon \to \frac{1}{2}$ inspired by the identifiability proof in Lemma \ref{lemma:generaliden}. Theorem 4 from \cite{diakonikolas2022robust} attains the optimal error for small $\varepsilon$\footnote{We refer the reader to our analysis for optimal breakdown point (Theorem \ref{thm:sparseoptbd}), which also recovers this result.}. Putting together the results for both the regimes of corruption gives us the guarantee of Theorem \ref{theorem:sparse}.
\begin{proof}
We overload notation like we did in previous proofs -- we will use $\mu^* = \sn x_i^*$. Defining $\delta \coloneqq 1 - 2\varepsilon$, we will prove the following: 
\begin{align*}
    \cA_{\text{sparse}} \sststile{O(t)}{w,x,v} \langle v, \mu - \mu^* \rangle \leqslant O\left( M^{1/t} \cdot \delta^{-1/t}\right).
\end{align*}
Then by taking any degree $O(t)$ pseudo-expectation satisfying $\cA_{\text{sparse}}$, we will get the desired guarantee by picking our estimator $\hat{\mu} := \pE[\mu]$ as discussed in Section \ref{sec:rounding}. 

The proof of  $\cA_{\text{sparse}} \sststile{O(t)}{w,x,v} \langle v, \mu - \mu^* \rangle \leqslant O\left( M^{1/t} \cdot \delta^{-1/t}\right)$ is essentially the same as the proof of Theorem \ref{theorem:general} with minor changes and we therefore defer it to Appendix \ref{app:proofs}.
\end{proof}

%% file: appendix.tex
\section{Information-Theoretic Lower Bounds}\label{sec:lower}

\subsection{Information Theoretic Limits of Robust Estimation}
\begin{lemma}[Limits of Robust Estimation \cite{diakonikolas2023algorithmic}]
    Let $\D_1$ and $\D_2$ be distributions such that $d_{TV}(\D_1, \D_2) \leqslant 2\varepsilon$ for $\varepsilon \in (0,\frac{1}{2})$. Without loss of generality, assume that the true distribution is $\D_1$. Then under $\varepsilon$-corruption, no algorithm can reliably distinguish between $\D_1$ and $\D_2$. 
\end{lemma}

\begin{proof}
    Notice that the adversary can corrupt the samples in a way that the algorithm observes samples from $\widetilde{\D} \coloneqq  \frac{1}{2}\left(\D_1 + \D_2\right)$ which is $\varepsilon$ far in statistical distance from both $\D_1$ and $\D_2$. This makes it information-theoretically impossible to distinguish between $\D_1$ and $\D_2$ in the above setting. 
\end{proof}

\subsection{Lower Bound for Bounded Moments}\label{lb:moments}

\begin{lemma}
    There exist two distributions $\D_1$ and $\D_2$ on $\R$ such that for $\varepsilon \in (0,\frac{1}{2})$
    \begin{enumerate}
        \item $d_{TV}(\D_1, \D_2) \leqslant 2\varepsilon$.
        \item $\D_1$ and $\D_2$ have bounded $k^{th}$ moments for $k \geqslant 2$ in the sense of Definition \ref{defn:certifiablebounded}. 
        \item $\left| \E_{X \sim \D_1}[X] - \E_{X \sim \D_2}[X] \right| \geqslant  \sqrt{k}\cdot \varepsilon^{1 - 1/k}\cdot{(1 - 2\varepsilon)^{-1/k}}$.
    \end{enumerate}
\end{lemma}
\begin{proof}
    Take $\D_1$ to be the distribution that outputs 0 with probability 1. Take $\D_2$ to be the distribution that outputs samples from $\D_1$ with probability $1 - 2\varepsilon$, and outputs $\sqrt{k}\cdot{\left(2\varepsilon  (1 - 2\varepsilon)\right)^{-1/k}}$  with probability $2\varepsilon$. Notice that $d_{TV}(\D_1, \D_2) \leqslant 2\varepsilon$. Furthermore we have
    \begin{align*}
        \left| \E_{X \sim \D_1}[X] - \E_{X \sim \D_2}[X] \right| = \left|0 - \frac{(2\varepsilon)^{1-1/k}\sqrt{k}}{(1-2\varepsilon)^{1/k}}\right| \geqslant \sqrt{k}\varepsilon^{1 - 1/k}\cdot(1 - 2\varepsilon)^{-1/k}
    \end{align*}
    $\D_1$ clearly has bounded $k^{th}$ central moments. Let $\mu_2$ be the mean of $\D_2$.
    \begin{align*}
     \E_{X \sim \D_2} \left[ \left(X - \mu_2\right)^k\right] &= \mathbb{E} \left[ \left.\left(X - \mu_2\right)^k \right| X = 0\right]\cdot\Pr[X = 0 ] \\
     &+ \mathbb{E}\left[\left.\left(X - \mu_2\right)^k \right| X = \frac{\sqrt{k}}{\left(2\varepsilon  (1 - 2\varepsilon)\right)^{1/k}} \right]\cdot\Pr\left[X =  \frac{\sqrt{k}}{\left(2\varepsilon  (1 - 2\varepsilon)\right)^{1/k}}\right]\\
     &= \frac{k^{k/2}(2\varepsilon)^{k-1}}{(1-2\varepsilon)}\cdot(1 - 2\varepsilon) + \frac{(1-2\varepsilon)^k k^{k/2}}{2\varepsilon(1-2\varepsilon)}\cdot2\varepsilon  \\
     &= k^{k/2}\cdot \left( (2\varepsilon)^{k-1} + (1-2\varepsilon)^{k-1}\right)\\
     &\leqslant k^{k/2}.
    \end{align*}
    where the final inequality follows from the fact that $\left( (2\varepsilon)^{k-1} + (1-2\varepsilon)^{k-1}\right) \leqslant 1 \ \forall \varepsilon \in (0,\frac{1}{2})$. 
\end{proof}

\subsection{Lower Bound for Gaussians}\label{lb:gauss}
\begin{lemma}
    There exists two distributions $\D_1$ and $\D_2$ on $\R$ such that for $C \leqslant \varepsilon < \frac{1}{2}$ for a sufficiently large constant $C > 0$ such that
    \begin{enumerate}
        \item $d_{TV}(\D_1, \D_2) \leqslant 2\varepsilon$.
        \item $\D_1$ and $\D_2$ are Gaussians with unit variance.
        \item $\left| \E_{X \sim \D_1}[X] - \E_{X \sim \D_2}[X] \right| \geqslant \Omega \left(\sqrt{\log \frac{1}{1 - 2\varepsilon}}\right)$.
    \end{enumerate}
\end{lemma}

\begin{proof}
 Take $\D_1$ to be the standard Gaussian $\cN(0,1)$. Take $\D_2$ to be the Gaussian $\cN(\mu,1)$, with $\mu > 0$ such that $d_{TV}(\D_1, \D_2) = 2\varepsilon$. We consider Tsybakov's version of the Bretagnolle–Huber inequality \cite{bretagnolle1978estimation, tsybakovnonparametric2009, canonne2022short} which relates the statistical distance and KL divergence, especially for large KL divergences (See Appendix \ref{app:facts} for more details).
\begin{align*}
    d_{TV}(P,Q) \leqslant 1 - \frac{1}{2}\exp\left(-\kldiv{P}{Q}\right)
\end{align*}
 Via a simple calculation, we have
\begin{align*}
    \kldiv{\cN(\mu_1, \sigma_1^2)}{ \cN(\mu_2, \sigma_2^2
    )} =\log \frac{\sigma_2}{\sigma_1} + \frac{\sigma_1^2 + (\mu_1 - \mu_2)^2}{2\sigma_2^2} - \frac{1}{2}
\end{align*}
Applying this inequality for our choice of $\D_1$ and $\D_2$, we get 
\begin{align*}
   2\varepsilon &\leqslant 1 - \frac{1}{2}\exp\left(-\frac{\mu^2}{2}\right)
   \end{align*}
   which after simplifying yields 
\begin{align*}
    \mu &\geqslant \sqrt{2\cdot \log\left(\frac{1}{2 - 4\varepsilon}\right)} \geqslant \sqrt{\log \left(\frac{1}{1 - 2\varepsilon}\right) - \log 2} 
\end{align*}
From which we conclude that 
\begin{align*}
    \mu \geqslant \Omega\left(\sqrt{\log \frac{1}{1-2\varepsilon}} \right).
\end{align*}
\end{proof}

\section{Preliminaries}
\input{preliminaries}
\section{Useful Facts and Additional Discussion}\label{app:facts}

\subsection{Useful Inequalities}
\begin{fact}
Pinkser's Inequality \cite{pinsker1964information}\\
For any two distributions $P, Q$ on a measurable space, we have
\begin{align*}
    d_{TV}(P, Q) \leqslant \sqrt{\frac{1}{2} \cdot \kldiv{P}{Q}}
\end{align*}
\end{fact}
\begin{fact}
    Bretagnolle–Huber (BH) Inequality \cite{bretagnolle1978estimation}\\
For any two distributions $P, Q$ on a measurable space, we have
\begin{align*}
    d_{TV}(P, Q) \leqslant \sqrt{1 - \exp\left(- \kldiv{P}{Q} \right)}
\end{align*}
\end{fact}
\begin{fact}
Tsybakov's version of the BH Inequality \cite{tsybakovnonparametric2009}\\
For any two distributions $P, Q$ on a measurable space, we have
\begin{align*}
    d_{TV}(P, Q) \leqslant 1 - \frac{1}{2} \cdot \exp\left(- \kldiv{P}{Q} \right)
\end{align*} 
\end{fact}
Pinkser's inequality is vacuous for $\kldiv{P}{Q} \geqslant  2 $. Whenever $\kldiv{P}{Q} \geqslant 2$, BH inequality and Tsybakov's version of the BH inequality \cite{tsybakovnonparametric2009} are much more sharper and also have the right asymptotic behavior (i.e., when $\kldiv{P}{Q} \to \infty$). We refer the reader to \cite{canonne2022short} for a detailed discussion.\\

\subsection{Statistical Distance and Overlap}\label{app:factsoltv}
For any two distributions $P, Q$ on $\R^d$, we have
\begin{align*}
     d_{TV}(P, Q) &= \sup_{A \subseteq \mathbb{R}^d} \Pr_{X \sim P}(X \in A) - \Pr_{X \sim Q}(X \in A) = P(S^+) - Q(S^+) \\&= \sup_{A \subseteq \mathbb{R}^d} \Pr_{X \sim Q}(X \in A) - \Pr_{X \sim P}(X \in A) = Q(S^-) - P(S^-)
\end{align*}
where $S^+ = \{x \in \mathbb{R}^d: p(x) > q(x)\}$ and $S^- = \mathbb{R}^d - S^+.$ 

Consider $r(x) = \min(p(x), q(x))$. Then we have overlap $\delta$
\begin{align*}
    \delta \coloneqq  \int_x r(x) dx &= \int_{x: p(x) < q(x)}p(x)dx + \int_{x: p(x) \geqslant q(x)}q(x)dx \\ &= 1 - \int_{x : p(x) \geqslant q(x)}(p(x) - q(x))dx = 1 - d_{TV}(P, Q)
\end{align*}
The consequence is that an overlap of $\delta$ between two densities is the same as a separation in statistical distance of $1 - \delta$, since overlap is precisely the intersection region between the two densities.\\

\subsection{Modeling Universal Constraints}
In the polynomial system $\cA$ (System \ref{eqn:finalsystem}), we utilized universal quantifiers within the system. In particular we used the following in our system 
\begin{align*}
    \forall v \in \R^d : \sn \langle x_i - \mu, v \rangle^k \leqslant \vnorm{v}^k \cdot k^{k/2}
\end{align*}
We note here that we assumed that the distributions we considered have \emph{certifiably bounded moments}. That is to say we have an SoS proof of the above inequality as well in variables $v$: 
\begin{align*}
    \sststile{O(k)}{v} \vnorm{v}^k \cdot k^{k/2} - \sn \langle x_i - \mu, v \rangle^k \geqslant 0
\end{align*}
We note that this is equivalent to searching for a positive semi-definite matrix $M \in \R^{(d+1)^{k/2} \times (d+1)^{k/2}}$ (in variables $\pxi$) such that 
\begin{align*}
    \vnorm{v}^k \cdot k^{k/2} - \sn \langle x_i - \mu, v \rangle^k = \left\langle (1, v)^{\otimes k/2}, M (1, v)^{\otimes k/2} \right\rangle
\end{align*}
where $(1,v)$ is the vector $v$ with 1 appended to it at the beginning. The above equality allows us to eliminate the variables $v$, similar to the bounded covariance case, and simply asks to search for a PSD matrix $M$ which only depends on variables $\pxi$. We can therefore add this positive semi-definite constraint to the program when computing the pseudo-expectation satisfying $\cA$. We refer the reader to \cite{fleming2019semialgebraic, steurer2021sos} for more information on transforming universal quantifiers into existential quantifiers in SoS.

\subsection{High Probability Statements}\label{app:hp}
Let $\uip \overset{\mathrm{iid}}{\sim} \D \subset \R^d$ and let $\D$ have mean $\mu^*$ and covariance $\Sigma$. We then have by cyclicity and linearity of Trace that
\begin{align*}
    \E \vnorm{\sn x_i^* - \mu^*}^2 &= \E \left[\Tr\left(\sn x_i^* - \mu^* \right) \left( \snj x_j^* - \mu^* \right)^T \right]\\
    &= \Tr \left[ \E \left[\frac{1}{n^2}\sum_{i,j}x_i^* {x_j^*}^T - \sn x_i^* {\mu^*}^T - \mu^* \left(\snj x_j^*\right)^T + \mu^* {\mu^*}^T \right] \right]\\
    &= \Tr\left[\frac{\E[xx^T]}{n} + \frac{n(n-1)}{n^2}\mu^* {\mu^*}^T - \mu^* {\mu^*}^T \right]\\
    &= \frac{1}{n} \Tr\left[\E[xx^T] - \mu^* {\mu^*}^T\right]= \frac{\Tr(\Sigma)}{n}
\end{align*}
For the different assumptions that we consider, we have the above quantity bounded by $O\left(\frac{d}{n}\right)$. We can thus use standard concentration inequalities to bound the closeness of sample mean to the true mean with probability at least $1 - \tau$. For our theorem statements we set $\tau = 0.01$. We refer the reader to \cite{zhu2022generalized} (Table 1 and corresponding proofs in the appendices therein) for precise statements in this regard.
We do not concern ourselves with the exact dependence of the sampling error on the confidence parameter $\tau$ -- there is a line of research that concerns itself with optimizing the dependence on $\tau$ and that is not the primary focus of this paper. We refer the reader to \cite{diakonikolas2023algorithmic} for more information in this regard.

\subsubsection{Bounded Covariance Case} 
\begin{lemma}[\cite{diakonikolas2017being, li2019lectures}]\label{lemma:boundedcovconc}
Let $\D$ be a distribution with mean $\mu$ and covariance $\Sigma$. Suppose $\Sigma \preccurlyeq I_d$. Let $\pxi \overset{\mathrm{iid}}{\sim} \D$. Let $\alpha \in [0, \frac{1}{2})$ and let $\tau > 0$. Then there exists universal constants $c, c'$ such that with probability $1 - \tau - \exp(-\Omega(\alpha \cdot n))$ such that there is a subset $S \subset [n]$, such that $|S| \geqslant (1 - \alpha) \cdot n$ and
\begin{align*}
    \vnorm{\hat{\mu} - \mu} \leqslant \frac{1}{1 - \alpha} \left(\sqrt{\frac{d}{n \tau}} + \sqrt{c \alpha} \right)
\end{align*}
\begin{align*}
    \vnorm{ \frac{1}{|S|} \sum_{i \in S}\left(x_i - \hat{\mu}\right)\left(x_i - \hat{\mu}\right)^T} \leqslant \frac{1}{1 - \alpha} \cdot \frac{d \log d + \log 2/ \tau}{nc'\alpha}
\end{align*}
where $\hat{\mu} = \frac{1}{|S|}\sum_{i \in S} x_i $
\end{lemma}
By appropriately choosing $c, c', \alpha, \tau$, with probability 0.99, we can
\begin{itemize}
    \item Bound the error of the sample mean of points in $S$ in estimating the true mean by $O\left(\sqrt{\frac{d}{n}}\right)$, and
    \item Ensure that the sample covariance over the set $S$ has largest eigenvalue at most $O\left({\frac{d \log d}{n}}\right)$
\end{itemize} 
whenever $d$ is sufficiently large. In particular, in our SoS proofs (especially the proof of Theorem \ref{thm:boundedcov}), we use $n = \Omega(d \log d)$ and assume that the largest eigenvalue of the empirical covariance over the uncorrupted samples is bounded by a constant.

\subsection{Sample Complexity Dependence on Corruption Rate}
It is common practice in robust mean estimation papers to report the sample complexity dependence as a polynomial in $d$ and  $\varepsilon^{-1}$. The polynomial in $\varepsilon^{-1}$ is taken in such a way that the sampling error and robust error are the same up to constant factors \footnote{The robust error is the error term that depends only on $\varepsilon$. The other term is sampling error.} This lends itself to the interpretation that in order to achieve vanishing estimation error, the number of samples required tends to infinity (since $\varepsilon^{-1} \to \infty$ as $\varepsilon \to 0$).
Throughout this paper we instead explicitly state both the robust and the sampling error terms and state sample complexity to reflect the number of samples required for our distributional assumptions to hold to the uniform distribution  over the uncorrupted samples.

\subsection{On taking moments that are a power of 2} We remark here that we can relax certain conditions that we make in our theorem statements. In particular, we have for indeterminate $v$ and $C > 0$ the following proof in SoS.
\begin{align*}
    \| v \|^2 = \frac{1}{C^{k-2}}.\|v \|^2 \cdot C^{k-2} = \frac{1}{C^{k-2}}.\|v \|^2 \cdot \left(C^2\right)^{(k-2)/2} \leqslant \frac{2}{k \cdot C^{k-2}}\left( \Vert v \Vert^k + \left( \frac{k}{2} - 1\right)\cdot C^k \right)
\end{align*}
where the inequality is due to SoS AM-GM with $t = \frac{k}{2}$. Picking $v = \del$, we get
\begin{align*}
    \err^2 \leqslant \frac{2}{k} \cdot \left( \frac{\err^k }{C^{k-2}} + \frac{k-2}{k} \cdot C^2\right)
\end{align*}
Finally, picking $C^k = k^{k/2}\cdot 2^k \cdot \delta^{-1}$ and using the bound on $\err^k$ in our SoS proof, we get the desired result for any $k$ that is a multiple of 4.

\subsection{Subgaussian Distributions}
Recently, \cite{diakonikolas2024sos} showed that all subgaussian distributions are in fact certifiably subgaussian. As a consequence, we get the optimal $O\left(\sqrt{\log \frac{1}{1-2\varepsilon}}\right)$ error even for subgaussian distributions in quasi-polynomial time.
\subsection{Comparison between Lemma \ref{lemma:ineff} and Theorem \ref{theorem:general}}\label{app:factscompare}

We have the following from Lemma \ref{lemma:ineff}:
\begin{align*}
    \cA \sststile{O(k)}{w,x,v} \left(\sn w_iw_i^* \right)^{k} \cdot \err^{2k} \leqslant 2^k \cdot k^{k/2} \cdot \left(\sn w_iw_i^* \right)^{k-1} \cdot \err^k
\end{align*}
The above statement effectively certifies the non-negativity of the following polynomial in SoS:
\begin{align*}
   \underbrace{\left(\sn w_iw_i^* \right)^{k-1}}_{(A)} \cdot \underbrace{\left( 2^k \cdot k^{k/2} \cdot \err^k - \left(\sn w_iw_i^* \right) \cdot \err^{2k} \right)}_{(B)} \geqslant 0
\end{align*}
Now, if we observe carefully, $(B)$ is precisely the polynomial whose non-negativity we certify in SoS in Theorem \ref{theorem:general}! By Lemma \ref{lemma:thm}, we already certify the non-negativity of $(A)$ above. 

In summary, the SoS proof of Lemma \ref{lemma:ineff} has a factorization of two terms, and each individual term has an SoS certificate of non-negativity. We remark here that it is generally \emph{not} the case that a SoS proof also has factors that are SoS themselves -- making the above polynomials quite interesting in their own right.

\section{Theorems and Proofs}\label{app:lemmaproof}
\subsection{Theorems demonstrating optimal breakdown point}

\begingroup 
\allowdisplaybreaks
\begin{theorem}[Optimal Breakdown for Bounded Covariance]\label{thm:optbdbcov}
Let $\uip \subset \R^d$ such that $\mus = \sn \xs_i$  and  $\Ss = \sn (\xs_i - \mus)(\xs_i - \mus)^T \preccurlyeq I_d$. 
Let $\{z_i\}_{i=1}^n$ be an $\varepsilon$-corruption of $\uip$ where $\varepsilon \in [0, \frac{1}{2})$. Let $\delta = 1 - 2\varepsilon$.  
Then
  \begin{align*}
      \cA_{w,x} \sststile{4}{w, x} \vnorm{\mu - \mu^*}^2 \leqslant O\left(   \frac{\varepsilon}{(1 - 2\varepsilon)^2}\right) = O\left( \frac{\varepsilon}{\delta^2}\right)
  \end{align*}
\end{theorem}
We remark that Theorem \ref{thm:optbdbcov} implies a breakdown point of $\frac{1}{2}$ for the estimator $\hat{\mu} = \pE[\mu]$ for a constant degree pseudo-expectation $\pE$ satisfying $\cA_{w,x}$ (See Section \ref{sec:rounding}). As $\varepsilon \to \frac{1}{2}$, the error \emph{diverges} and is $O\left(\delta^{-1}\right)$ (after taking square root). 
When $\varepsilon$ is sufficiently small, the error is optimal up to constant factors while the same is not true when $\varepsilon \to \frac{1}{2}$ (respectively $\delta \to 0$).

We now prove Theorem \ref{thm:optbdbcov}. 
\begin{proof}
\begin{align*}
       \vnorm{\mu - \mu^*}^2 &= \langle \mu - \mu^*, \mu - \mu^* \rangle = \sn \langle x_i - x_i^*, \mu - \mu^* \rangle 
       = \sn (1 - w_iw_i^*) \langle x_i - x_i^*, \mu - \mu^* \rangle  \\
       &= \sn(1-w_iw_i^*) \langle x_i - \mu - x_i^* + \mu^* + \mu - \mu^*, \mu - \mu^* \rangle \\
       &= \sn(1-w_iw_i^*)\langle x_i - \mu, \mu - \mu^* \rangle - \sn (1 - w_iw_i^*) \langle x_i^* - \mu^*, \mu - \mu^* \rangle\\
       &+ \sn(1-w_iw_i^*) \Vert \mu - \mu^*\Vert^2 \\
       & \leqslant \sn(1-w_iw_i^*)\langle x_i - \mu, \mu - \mu^* \rangle - \sn (1 - w_iw_i^*) \langle x_i^* - \mu^*, \mu - \mu^* \rangle\\
       &+ 2\varepsilon \vnorm{\mu - \mu^*}^2 
\end{align*}
Above, we used Lemma \ref{lemma:thm} for the third equality and Lemma \ref{lemma:thm} for the inequality.
 Rearranging and squaring we have
 \begin{align*}
     (1-2\varepsilon)^2 \Vert \mu - \mu^* \Vert^4 &\leqslant \left(\sn(1-w_iw_i^*)\langle x_i - \mu, \mu - \mu^* \rangle - \sn (1 - w_iw_i^*) \langle x_i^* - \mu^*, \mu - \mu^* \rangle \right)^2
 \end{align*}
Notice that we have already optimized the proof for optimal breakdown point above. This is a general approach that we use in both Theorems \ref{thm:optbdmom} and \ref{thm:sparseoptbd}. The right hand side can be bounded above by
 \begin{align*}
     &\left(\sn(1-w_iw_i^*)\langle x_i - \mu, \mu - \mu^* \rangle - \sn (1 - w_iw_i^*) \langle x_i^* - \mu^*, \mu - \mu^* \rangle \right)^2 \\ 
     &\leqslant 2 \left(\sn(1-w_iw_i^*)\langle x_i - \mu, \mu - \mu^* \rangle\right)^2 + 2 \left( \sn (1 - w_iw_i^*) \langle x_i^* - \mu^*, \mu - \mu^* \rangle \right)^2
 \end{align*}
 by SoS triangle inequality. 
 We now focus on the first term on the right hand side of the above equation. By SoS Cauchy-Schwarz, we have 
\begin{align*}
    \left(\sn(1-w_iw_i^*)\langle x_i - \mu, \mu - \mu^* \rangle\right)^2 &\leqslant \left(\sn (1 - w_iw_i^*)^2\right)\left(\sn \langle x_i - \mu, \mu - \mu^* \rangle^2 \right) \\
    &= \left(\sn (1 - w_iw_i^*)\right)\left(\sn \langle x_i - \mu, \mu - \mu^* \rangle^2 \right)\\
    &\leqslant 2\varepsilon \cdot \vnorm{\mu - \mu^*}^2 
\end{align*}
The equality above is due to Lemma \ref{lemma:thm} and the final inequality follows from bounded covariance constraint. Similarly, we have
\begin{align*}
    \left( \sn (1 - w_iw_i^*) \langle x_i^* - \mu^*, \mu - \mu^* \rangle \right)^2 \leqslant 2\varepsilon \cdot \vnorm{\mu - \mu^*}^2 
\end{align*}
Putting all the pieces together, we have 
\begin{align*}
    (1 - 2\varepsilon)^2 \err^4 \leqslant 8 \varepsilon \cdot \err^2
\end{align*}
By SoS Cancellation we get that 
\begin{align*}
    \err^2 \leqslant \frac{8 \varepsilon}{(1 - 2 \varepsilon)^2}
\end{align*}
which proves the Theorem. 
\end{proof}
\endgroup

\textbf{Remark}: Prior works based on SoS focus on the regime when $\varepsilon$ is sufficiently small and bounded away from $\frac{1}{2}$ and their proofs are inherently not optimized when $\varepsilon$ is large and close to $\frac{1}{2}$. In the above proof, we already optimized it to demonstrate optimal breakdown point. To this end, it was sufficient to get a factor of $2\varepsilon$ on the right hand side in the steps before rearranging and squaring. We similarly replicate this proof for the other cases with different distributional assumptions such as certifiably bounded higher moments and for the problem of sparse mean estimation in Theorems \ref{thm:optbdmom} and \ref{thm:sparseoptbd}.

\begin{theorem}[Optimal Breakdown Point for Bounded Moments]\label{thm:optbdmom}
      Under the conditions of Theorem \ref{theorem:general}, and for $\delta \coloneqq 1 - 2\varepsilon$ we have
    \begin{align*}
        \Vert \hat{\mu} - \mus \Vert \leqslant O\left( \frac{{\varepsilon}^{1-1/k} \cdot \sqrt{k}}{\delta}\right)
    \end{align*}
    with probability 0.99.
\end{theorem}
\begin{proof}We overload notation like we did in the proof of Theorem \ref{thm:boundedcov} -- we will use $\mu^* = \sn x_i^*$.
    As with the proofs from before, it will be enough to show that 
    \begin{align*}
       \cA \sststile{O(k)}{w,x} \Vert \mu - \mu^* \Vert^2  \leqslant O\left( \frac{{\varepsilon}^{2-2/k} \cdot k}{\delta^2}\right)
    \end{align*}
    Using the same idea as in Theorem \ref{thm:optbdbcov}, we have
    \begin{align*}
     (1-2\varepsilon) \Vert \mu - \mu^* \Vert^2 &\leqslant \left(\sn(1-w_iw_i^*)\langle x_i - \mu, \mu - \mu^* \rangle - \sn (1 - w_iw_i^*) \langle x_i^* - \mu^*, \mu - \mu^* \rangle \right)
 \end{align*}
 Raising both sides to the power of $k$ we have
 \begin{align*}
     (1-2\varepsilon)^{k} \Vert \mu - \mu^* \Vert^{2k} &\leqslant \left(\sn(1-w_iw_i^*)\langle x_i - \mu, \mu - \mu^* \rangle - \sn (1 - w_iw_i^*) \langle x_i^* - \mu^*, \mu - \mu^* \rangle \right)^k
 \end{align*}
By applying SoS Triangle Inequality on the RHS we get
\begin{align*}
    &\left(\sn(1-w_iw_i^*)\langle x_i - \mu, \mu - \mu^* \rangle - \sn (1 - w_iw_i^*) \langle x_i^* - \mu^*, \mu - \mu^* \rangle \right)^k \leqslant \\
    & \ 2^{k-1}  \cdot \left(\left(\sn(1-w_iw_i^*)\langle x_i - \mu, \mu - \mu^* \rangle \right)^k + \left(\sn (1 - w_iw_i^*) \langle x_i^* - \mu^*, \mu - \mu^* \rangle \right)^k \right)
\end{align*}
We will focus on the first term in the RHS above, as the proof for the second term is analogous. From SoS Holder and Lemma \ref{lemma:thm}, we have
\begin{align*}
   \left(\sn(1-w_iw_i^*)\langle x_i - \mu, \mu - \mu^* \rangle \right)^k &\leqslant \left( \sn(1-w_iw_i^* \right)^{k-1} \left( \sn \langle x_i - \mu, \del \rangle^k \right)\\
   &\leqslant (2\varepsilon)^{k-1} \cdot \Vert \del \Vert^k \cdot k^{k/2}
\end{align*}
where the final inequality is because of our assumptions on the moments and Lemma \ref{lemma:thm}.

Putting it together we get
\begin{align*}
    (1-2\varepsilon)^{k} \Vert \mu - \mu^* \Vert^{2k} \leqslant (2\varepsilon)^{k-1} \cdot 2^k \cdot k^{k/2} \err^k
\end{align*}
Using SoS Cancellation we get
\begin{align*}
    \err^k \leqslant \frac{2^{2k-1}\cdot \varepsilon^{k-1}\cdot k^{k/2}}{\delta^k}
\end{align*}
By SoS Square Root we get
\begin{align*}
    \cA \sststile{O(k)}{w,x} \err^2 = O\left( \frac{{\varepsilon}^{2-2/k} \cdot k}{\delta^2}\right)
\end{align*}
which proves the Theorem.
\end{proof}

\begingroup
\allowdisplaybreaks
\begin{theorem}[Optimal Breakdown Point for Sparse Mean Estimation]\label{thm:sparseoptbd}
     Under the conditions of Theorem \ref{theorem:sparse} and for $\delta \coloneqq 1 - 2\varepsilon$, we have
    \begin{align*}
        \vnorm{\hat{\mu} - \mu^*}_{2,k} \leqslant O\left( \frac{M^{1/t} \cdot \varepsilon^{1-1/t}}{\delta}\right)
    \end{align*}
    with probability 0.99.
\end{theorem}
\begin{proof}
We overload notation like we did in the proof of Theorem \ref{thm:boundedcov} -- we will use $\mu^* = \sn x_i^*$. Similar to the proof of Theorems \ref{thm:optbdbcov} and \ref{thm:optbdmom}, we will show that 
$\cA_{\text{sparse}} \sststile{O(t)}{w,x,v} \langle v, \mu - \mu^* \rangle \leqslant O\left( {M^{1/t} \cdot \varepsilon^{1-1/t}}\cdot{\delta^{-1}}\right)$. This will be sufficient by the discussion in Section \ref{sec:rounding}.
Indeed we have, for any $v \in S_k$
\begin{align*}
    \langle v, \mu - \mu^* \rangle = \sn  (1 - w_iw_i^*) \langle v, x_i - x_i^* \rangle
\end{align*}
where we used Lemma \ref{lemma:thm}. This can further be written as 
\begin{align*}
  \sn  (1 - w_iw_i^*) \langle v, x_i - x_i^* \rangle &= \sn(1 - w_iw_i^*)\langle v, x_i - \mu \rangle - \sn (1 - w_iw_i^*) \langle v, x_i^* - \mu^* \rangle \\ 
  & \ \ \ \ + \sn (1 - w_iw_i^*)\langle v, \mu - \mu^* \rangle \\
  &\leqslant \sn(1 - w_iw_i^*)\langle v, x_i - \mu \rangle - \sn (1 - w_iw_i^*) \langle v, x_i^* - \mu^* \rangle \\
  & \ \ \ \ + 2\varepsilon \cdot \langle v, \mu - \mu^* \rangle
\end{align*}
where we used Lemma \ref{lemma:thm}. Rearranging we get
\begin{align*}
    (1 - 2\varepsilon)\langle v, \mu - \mu^* \rangle = \delta \langle v, \mu - \mu^* \rangle \leqslant  \sn(1 - w_iw_i^*)\langle v, x_i - \mu \rangle - \sn (1 - w_iw_i^*) \langle v, x_i^* - \mu^* \rangle \\
\end{align*}
Raising both sides to the power of $t$, we have
\begin{align*}
    \delta^t  \langle v, \mu - \mu^* \rangle^t &\leqslant \left(\sn(1 - w_iw_i^*)\langle v, x_i - \mu \rangle - \sn (1 - w_iw_i^*) \langle v, x_i^* - \mu^* \rangle\right)^t \\
    &\leqslant 2^{t-1}\cdot \left[ \left(\sn(1 - w_iw_i^*)\langle v, x_i - \mu \rangle \right)^t + \left(\sn (1 - w_iw_i^*) \langle v, x_i^* - \mu^* \rangle\right)^t \right] \\
    &\leqslant 2^{t-1} \cdot \left(\sn(1-w_iw_i^*)\right)^{t-1}\cdot \left(\sn \langle v, x_i - \mu \rangle^t \right) \\
    & \ \ \ \ + 2^{t-1} \cdot \left(\sn(1-w_iw_i^*)\right)^{t-1}\cdot \left(\sn \langle v, x_i^* - \mu^* \rangle^t \right) \\
    &\leqslant 2^{t-1}\cdot (2\varepsilon)^{t-1}\cdot \left(\sn \langle v, x_i - \mu \rangle^t + \sn \langle v, x_i^* - \mu^* \rangle^t \right)
\end{align*}
where we used SoS Triangle Inequality, then SoS \Hoelder's inequality and finally Lemma \ref{lemma:thm}. Squaring, we get
\begin{align*}
   \delta^{2t}  \langle v, \mu - \mu^* \rangle^{2t} &\leqslant  2^{4t-4}\cdot \varepsilon^{2t-2} \cdot \left(\sn \langle v, x_i - \mu \rangle^t + \sn \langle v, x_i^* - \mu^* \rangle^t \right)^2 \\
   &\leqslant 2^{4t-3} \cdot \varepsilon^{2t-2}\left[ \left(\sn \langle v, x_i - \mu \rangle^t \right)^2 + \left(\sn \langle v, x_i^* - \mu^* \rangle^t \right)^2\right]\\
   &\leqslant 2^{4t-2} \cdot \varepsilon^{2t-2} \cdot M^2
\end{align*}
where we used SoS Triangle Inequality again and the boundedness assumptions.
Taking SoS Square Root,
\begin{align*}
    \cA_{\text{sparse}} \sststile{O(t)}{w,x,v} \langle v, \mu - \mu^* \rangle \leqslant \frac{2^{2 - 1/t}\cdot \varepsilon^{1 - 1/t}\cdot M^{1/t}}{\delta} = O\left( \frac{M^{1/t} \cdot \varepsilon^{1-1/t}}{\delta}\right).
\end{align*}
which proves the desired statement.
\end{proof}
\endgroup
\subsection{Missing Proofs}\label{app:proofs}

\begingroup
\allowdisplaybreaks
\subsubsection*{ Proof of Lemma \ref{lemma:thm}}
    \begin{proof}\label{prooflemmathreefour}
    We provide the proofs for the three parts individually. Firstly observe that for the first part that
    \begin{align*}
        w_iw_i^* \cdot x_i = w_i^* (w_i \cdot x_i) = w_i^* (w_i \cdot z_i) = w_i (w_i^* \cdot z_i) = w_i w_i^* \cdot x_i^* 
    \end{align*}
     We now have for part 2 that
     \begin{align*}
        (1 - w_iw_i^*)^2 = 1 + (w_iw_i^*)^2 - 2 w_iw_i^* = 1 + w_iw_i^* - 2 w_iw_i^* = 1 - w_iw_i^*
    \end{align*}
    Finally for part 3 we begin by proving the following fact.
    \begin{align*}
        \left\{w^2 = w\right\} \sststile{2}{w} w \leqslant 1
    \end{align*}
    Observe that 
    \begin{align*}
        w = 1\cdot w \leqslant \frac{1}{2} + \frac{{w}^2}{2} = \frac{1}{2} + \frac{w}{2}
    \end{align*}
    from which we can conclude that $w \leqslant 1$. In the above, we used the fact that $(a-b)^2 \geqslant 0 \ \forall a, b \in \R$. \\
    Therefore we have $ \cA_{w,x} \sststile{2}{w} 1 - w_i \geqslant 0 $ using which we have
    \begin{align*}
         \cA_{w,x} \sststile{2}{w} (1 - w_i^*)\cdot(1 - w_i) \geqslant 0 
    \end{align*}
    since $w_i^*$ is a constant for our system.
    Taking summation we have 
    \begin{align*}
        \cA_{w,x} \sststile{2}{w} \sn \left(1 - w_i^* - w_i + w_iw_i^*\right) \geqslant 0 
    \end{align*}
    Now we know $\sn w_i \geqslant 1 - \varepsilon$ and $\sn w_i^* \geqslant 1 - \varepsilon$. Using this we get 
    \begin{align*}
       \cA_{w,x} \sststile{2}{w} \sn (1 - w_iw_i^*) &\leqslant 2 - \sn w_i - \sn  w_i^* \leqslant 2 - (1 - \varepsilon) - (1 - \varepsilon) = 2\varepsilon
    \end{align*}
    By rearranging the terms we thus get $\cA_{w,x} \sststile{2}{w} \sn w_iw_i^* \geqslant 1 - 2\varepsilon = \delta$. 
\end{proof}
\endgroup

\begingroup
\allowdisplaybreaks
\subsubsection*{ Proof of Lemma \ref{lemma:ineff}}
    \begin{proof}\label{prooflemmathreeseven}
    Overloading notation like we did in the proof of Theorem \ref{thm:boundedcov} and using Lemma \ref{lemma:thm} we have
    \begin{align*}
        \left(\sn w_iw_i^* \right) \cdot \langle \del, v \rangle &= \left(\sn w_iw_i^* \right) \cdot \langle \del, v \rangle + \sn w_iw_i^*  \cdot \langle x_i^* - x_i, v \rangle \\
         &= \sn w_iw_i^* \cdot \langle \mu - x_i, v \rangle + \sn w_iw_i^* \cdot \langle x_i^* - \mu^*, v \rangle
    \end{align*}
    where the first inequality is due to Lemma \ref{lemma:thm}.
    
    We now raise both sides to the power of $k$ to obtain
    \begin{align*}
        \left(\sn w_iw_i^* \right)^k \cdot \langle \del, v \rangle^k &= \left(\sn w_iw_i^* \cdot \langle \mu - x_i, v \rangle + \sn w_iw_i^*  \cdot \langle x_i^* - \mu^*, v \rangle \right)^k \\
        &\leqslant 2^{k-1}\left[\left(\sn w_iw_i^* \cdot \langle \mu - x_i, v \rangle\right)^k +  \left( \sn w_iw_i^* \cdot \langle x_i^* - \mu^*, v \rangle \right)^k \right]
    \end{align*}
    where the inequality follows from SoS Triangle Inequality. Now using the booleanity of $w_iw_i^*$, we apply SoS \Hoelder to the two terms above. Therefore we have
    \begin{align*}
        \left(\sn w_iw_i^* \cdot \langle \mu - x_i, v \rangle\right)^k  &\leqslant \left(\sn w_iw_i^* \right)^{k-1}\cdot \left(\sn \langle x_i - \mu, v \rangle^k \right) \\
        &\leqslant \left(\sn w_iw_i^* \right)^{k-1}\cdot \Vert v \Vert^k \cdot k^{k/2}
    \end{align*}
where the final inequality follows from certifiably bounded central moments. Similarly we can show 
\begin{align*}
    \left( \sn w_iw_i^* \cdot \langle x_i^* - \mu^*, v \rangle \right)^k \leqslant  \left(\sn w_iw_i^* \right)^{k-1}\cdot \Vert v \Vert^k \cdot k^{k/2}
\end{align*}
Putting the pieces together we obtain
\begin{align*}
     \left(\sn w_iw_i^* \right)^k \cdot \langle  \del, v \rangle^k  \leqslant 2^k \cdot \left(\sn w_iw_i^* \right)^{k-1} \cdot \Vert v \Vert^k \cdot k^{k/2}.
\end{align*}
Picking $v = \del$ proves the Lemma.
\end{proof}
\endgroup 

\subsubsection*{Proof of Theorem \ref{theorem:sparse}}\label{proof:theoremsparse}
\begin{proof}
    As discussed earlier, it is sufficient to show 
    \begin{align*}
        \cA_{\text{sparse}} \sststile{O(t)}{w,x,v} \langle v, \mu - \mu^* \rangle \leqslant O\left( M^{1/t} \cdot \delta^{-1/t}\right)
    \end{align*}
    Indeed we have
    \begin{align*}
        \delta \cdot \langle v, \mu - \mu^* \rangle^t &\leqslant \left(\sn w_iw_i^* \right)\cdot \langle v, \mu - \mu^* \rangle^t \\
        &= \sn \left(w_iw_i^*\right)^t \cdot \langle v, \mu - \mu^* \rangle^t \\
        &= \sn \left( w_iw_i^* \cdot \langle v, \mu - \mu^* \rangle\right)^t
    \end{align*}
    where in the first inequality we used Lemma \ref{lemma:thm} and the first equality follows from the booleanity of $w_iw_i^*$ (Lemma \ref{lemma:thm}). We now have by Lemma \ref{lemma:thm} that
    \begin{align*}
        \sn \left( w_iw_i^* \cdot \langle v, \mu - \mu^* \rangle\right)^t &= \sn \left( w_iw_i^* \cdot \langle v, \mu - \mu^* \rangle + w_iw_i^* \cdot \langle v, x_i^* - x_i \rangle\right)^t \\
        &= \sn \left(w_iw_i^* \cdot \langle v, \mu - x_i \rangle + w_iw_i^* \cdot \langle v, x_i^* - \mu^* \rangle\right)^t\\
        &\leqslant 2^{t-1}\cdot\left(\sn\left(w_iw_i^* \cdot \langle v, \mu - x_i \rangle\right)^t + \sn \left(w_iw_i^* \cdot\langle v, x_i^* - \mu^* \rangle \right)^t\right)\\
        &= 2^{t-1} \cdot \left(\sn(w_iw_i^*)^t \cdot\langle v, \mu - x_i \rangle^t + \sn (w_iw_i^*)^t \cdot \langle v, x_i^* - \mu^* \rangle^t\right)\\
        &\leqslant 2^{t-1} \cdot \left(\sn \langle v, \mu - x_i \rangle^t + \sn \langle v, x_i^* - \mu^* \rangle^t\right) 
    \end{align*}
    where the first inequality is due to SoS Triangle Inequality, the final inequality is due to the fact that $\forall i \in [n]: w_iw_i^* \leqslant 1$ (See Proof of Lemma \ref{lemma:thm} in Appendix \ref{app:lemmaproof}). Squaring and applying SoS Triangle Inequality, we have
    \begin{align*}
        \delta^2 \cdot \langle v, \mu - \mu^* \rangle^{2t} &\leqslant 2^{2t-2} \cdot 2 \left[\left(\sn \langle v, \mu - x_i \rangle^t \right)^2 + \left(\sn \langle v, x_i^* - \mu^* \rangle^t\right)^2 \right] \\
        &\leqslant 2^{2t} \cdot M^2.
    \end{align*}
    where the final inequality follows from the moment boundedness assumption in the sense of Definition \ref{defn:sparsebounded}.
    Taking SoS Square root we have
    \begin{align*}
        \cA_{\text{sparse}} \sststile{O(t)}{w,x,v} \langle v, \mu - \mu^* \rangle \leqslant \frac{2\cdot M^{1/t}}{\delta^{1/t}} = O\left(\frac{M^{1/t}}{\delta^{1/t}}\right).
    \end{align*}
\end{proof}
\section{ Robust Gaussian Mean Estimation}\label{app:gauss}
In this section, we discuss applying the techniques from \cite{kothari2022polynomial} to our problem when $\varepsilon$ is large and close to the breakdown point. In particular, we will specialize Lemma 22 in \cite{kothari2022polynomial} for the case of robust Gaussian mean estimation with identity covariance. 
Consider the following polynomial system in variables $\pxi \in \R^d$ and $\pwi \in \R$. Consistent with before, let $\ip$ be an $\varepsilon$-corruption of $\uip \sim \cN(\mu^*, I_d)$.

\begin{align}\label{eqn:gausskmz}
     \cA_{\text{Gauss}} \coloneqq \left\{
    \begin{array}{l}
        \forall i \in [n]: \ w_i^2 = w_i   \\
        \forall i \in [n]: \ w_i\cdot (z_i - x_i) = 0\\
        \sn w_i \geqslant (1 - \varepsilon) \\
        \mu = \sn x_i \\
        \sn (x_i - \mu)(x_i - \mu)^T \preccurlyeq \left(1 +\tau \right)\cdot I_d
        \end{array}
    \right\}
\end{align}
In the above system the slack term of $\tau > 0$ takes the value of $O\left(\varepsilon \log 1/\varepsilon\right)$ in \cite{kothari2022polynomial}. Overloading notation let $\mu^*$ be the sample mean of the uncorrupted samples $\uip$. Further we have $\sn(x_i^* - \mu^*)(x_i^* - \mu^*)^T \preccurlyeq (1 + \tau)\cdot I_d$.

We interpret $\cA_{\text{Gauss}}$ as finding a \emph{specific} bounded covariance empirical distribution from corrupted samples. Note that this distinction is important as hoping to estimate the mean of a Gaussian with \emph{any} bounded covariance distribution might be too much to ask for. We provide lower bounds (Lemmas \ref{lemma:kmzlbdel} and  \ref{lemma:kmzlbeps}) which essentially argue that there exists a specific bounded covariance distribution and a Gaussian with identity covariance, such that 
\begin{enumerate}
    \item When $\varepsilon$ is large, these distributions are $2\varepsilon$ far away in statistical distance and have means $\sqrt{\frac{1}{1-2\varepsilon}}$ apart and 
    \item When $\varepsilon$ is small, these distributions are $\varepsilon$ far away in statistical distance and have means $\varepsilon\sqrt{\log \frac{1}{\varepsilon}} $ apart.
\end{enumerate}
These lower bounds rule out any algorithm that use only specific bounded covariance conditions such as the one in $\cA_{\text{Gauss}}$ to robustly estimate the mean of a Gaussian with good error, especially when $\varepsilon$ is large. Our lower bound when $\varepsilon$ is small also shows that \cite{kothari2022polynomial}'s result is tight.

\subsection{Lower Bounds for Bounded Covariance and Gaussian}
\begin{lemma}[Large $\varepsilon$]\label{lemma:kmzlbdel}
      There exist two distributions $\D_1$ and $\D_2$ on $\R$ such that for $\varepsilon \in (0,\frac{1}{2})$ sufficiently large
    \begin{enumerate}
        \item $d_{TV}(\D_1, \D_2) \leqslant 2\varepsilon$.
        \item $\D_1$ is $\cN(0,1)$ and $\D_2$ has variance bounded from above by $1 + O(\delta)$, where $\delta = 1 - 2\varepsilon$.
        \item $\left| \E_{X \sim \D_1}[X] - \E_{X \sim \D_2}[X] \right| \geqslant \sqrt{\frac{1}{1-2\varepsilon}}$.
    \end{enumerate}
\end{lemma}
\begin{proof}
     Let $\D_2$ to be the distribution that outputs samples from $\D_1$ with probability $1-2\varepsilon$, and outputs $(1-2\varepsilon)^{-1/2}\cdot (2\varepsilon)^{-1}$  with probability $2\varepsilon$. Notice that $d_{TV}(\D_1, \D_2) \leqslant 2\varepsilon$ and that the mean of $\D_2$ is $\mu_2 = (1-2\varepsilon)^{-1/2}$. 
     
     We have the variance of $\D_2$
     \begin{align*}
\E_{X \sim \D_2}[\left(X - \mu_2\right)^2] &=\E\left[(X-\mu_2)^2 | X \sim \cN(0,1)\right]\Pr\left[X \sim \cN(0,1) \right] + \\ 
         &\E\left[(X-\mu_2)^2 \left.\right| X = (1-2\varepsilon)^{-1/2}\cdot (2\varepsilon)^{-1}\right]\Pr\left[X = (1-2\varepsilon)^{-1/2}\cdot (2\varepsilon)^{-1} \right]\\
         &= \left( 1 + \frac{1}{1-2\varepsilon}\right) \cdot (1 - 2\varepsilon) + \left( (1-2\varepsilon)^2\cdot (1-2\varepsilon)^{-1}\cdot (2\varepsilon)^{-2} \right)\cdot 2\varepsilon \\
         &= 1 + 1 - 2\varepsilon + (1 - 2\varepsilon) \cdot (2\varepsilon)^{-1} =  1 + (1 - 2\varepsilon)\cdot\left(1 + \frac{1}{2\varepsilon}\right) \\
         &\leqslant 1 + (1 - 2\varepsilon)\cdot 3 \leqslant 1 + O(\delta).
     \end{align*}
    We note that we get the inequality by taking $\varepsilon > 1/4$. Observe that if the input to the program $\cA_{\text{Gauss}}$ is $\frac{1}{2}\left(\D_1 + \D_2\right)$, it is an $\varepsilon$ corruption of $\D_1$, and the $\varepsilon$ corrupted samples will simply satisfy the constraints, essentially allowing the algorithm to output $\frac{1}{2}\sqrt{\frac{1}{1-2\varepsilon}}$ as samples from $\frac{1}{2}\left(\D_1 + \D_2\right)$ are feasible for $\cA_{\text{Gauss}}$. Moreover, we show that in the large $\varepsilon$ regime, even when asking the bounded covariance distribution to have variance at most $1 + O(\delta)$, \emph{stronger} than the analogous $1 + O(\delta \log \frac{1}{\delta})$ as in \cite{kothari2022polynomial} still proves insufficient for obtaining optimal error rates.
\end{proof}

\begin{lemma}[Tightness of $\cite{kothari2022polynomial}$] \label{lemma:kmzlbeps}
      There exist two distributions $\D_1$ and $\D_2$ on $\R$ such that for $\varepsilon \in (0,\frac{1}{2})$ sufficiently small
    \begin{enumerate}
        \item $d_{TV}(\D_1, \D_2) \leqslant \varepsilon$.
        \item $\D_1$ is $\cN(0,1)$ and $\D_2$ has variance bounded from above by $1 + \varepsilon \log \frac{1}{\varepsilon}$.
        \item $\left| \E_{X \sim \D_1}[X] - \E_{X \sim \D_2}[X] \right| \geqslant  \varepsilon \sqrt{\log \frac{1}{\varepsilon}}$.
    \end{enumerate}
\end{lemma}
\begin{proof}
     Let $\D_2$ to be the distribution that outputs samples from $\D_1$ with probability $1 - \varepsilon$, and outputs $\sqrt{\log\frac{1}{\varepsilon}}$  with probability $\varepsilon$. Notice that $d_{TV}(\D_1, \D_2) \leqslant \varepsilon$ and that the mean of $\D_2$ is $\mu_2 = \varepsilon \cdot \sqrt{\log \frac{1}{\varepsilon}}$. We have the variance of $\D_2$
     \begin{align*}
\E_{X \sim \D_2}[\left(X - \mu_2\right)^2] &=\E\left[(X-\mu_2)^2 | X \sim \cN(0,1)\right]\Pr\left[X \sim \cN(0,1) \right] + \\ 
         &\E\left[(X-\mu_2)^2 \left.\right| X = \sqrt{\log\frac{1}{\varepsilon}}\right]\Pr\left[X = \sqrt{\log\frac{1}{\varepsilon}} \right]\\
         &= \left( 1 + \varepsilon^2\log 1/\varepsilon\right) \cdot (1 - \varepsilon) + (1 - \varepsilon)^2 \cdot \log 1/\varepsilon \cdot \varepsilon \\
         &= (1-\varepsilon) + \varepsilon \cdot (1 - \varepsilon) \cdot \left( \varepsilon \cdot  \log 1/\varepsilon + (1 - \varepsilon) \cdot \log 1/\varepsilon \right) \\ 
         &\leqslant 1 + \varepsilon \cdot (1 - \varepsilon) \cdot  \log 1/\varepsilon \leqslant 1 + \varepsilon \log \frac{1}{\varepsilon}.
     \end{align*}
     where the inequality was because $(1 - \varepsilon) \leqslant 1.$
     In this case, the adversary can simply generate samples from $\D_2$. The corrupted samples will satisfy the constraints of $\cA_{\text{Gauss}}$ and the algorithm will output $\varepsilon \sqrt{\log \frac{1}{\varepsilon}}$ as samples from $\D_2$ are feasible for $\cA_{\text{Gauss}}$.
\end{proof}

In both lower bounds above, the adversary can essentially produce corrupted samples that are feasible for the polynomial system, and that is key to the $O\left(\delta^{-1/2}\right)$ barrier.

\subsection{Inefficient Robust Estimator for the Gaussian Mean}
\begin{theorem}[Inefficient Estimator for the Gaussian Mean]\label{thm:zjs}
(Rephrasing of Theorem 4 from \cite{zhu2020does}): 
Let $\D$ be $\cN(\mu^*, I_d)$. 
Let $\delta = 1 - 2\varepsilon$.
Then for any $p$ such that $d_{TV}(p, \D) \leqslant \varepsilon$, there is an estimator $\hat{\mu}(p)$ with breakdown point $\frac{1}{2}$ such that
\begin{align*}
    \vnorm{\hat{\mu}(p) - \mu^*} \leqslant O\left(\sqrt{\log \frac{1}{\delta}}\right).
\end{align*}

Note that we only have access to $p$ and not $\D$.
\end{theorem}
\begin{proof}
    The error obtained in \cite{zhu2020does} is 
    \begin{align*}
        \vnorm{\hat{\mu}(p) - \mu^*} \leqslant 2 \cdot h^{-1}\left(\frac{1}{2} - \varepsilon\right)
    \end{align*}
    where
    \begin{align*}
        h(t) \coloneqq \sup_{\Vert v \Vert \leqslant 1} \Pr_{X \sim \D}\left[\langle v, X - \mu \rangle > t \right].
    \end{align*}
    and $h^{-1}$ is the generalized inverse.
    We show that when $\varepsilon \to \frac{1}{2}$, we can bound $2 \cdot h^{-1}\left(\frac{1}{2} - \varepsilon\right)$ from above by $O\left(\sqrt{\log \frac{1}{\delta}}\right)$. We have for $\D = \cN (\mu^*, I_d)$,
    \begin{align*}
    \forall t \geqslant 0 \ \ h(t) &= \sup_{\Vert v \Vert \leqslant 1} \Pr_{X \sim\cN(\mu^*, I)}\left[\langle v, X - \mu^* \rangle > t \right]
    =\sup_{\Vert v \Vert \leqslant 1}\Pr_{X' \sim\cN(0, \Vert v \Vert^2)}[X'>t]\\
    &= \sup_{\Vert v \Vert \leqslant 1}\Pr_{X'' \sim\cN(0,1)}\left[X'' > \frac{t}{\Vert v \Vert}\right]
    = \Pr_{X \sim\cN(0,1)} [X > t].
\end{align*}
For us $\delta = 1 - 2\varepsilon$. Using this we have estimation error
\begin{align*}
    2 \cdot h^{-1}\left(\frac{1}{2} - \varepsilon\right) = 2 \cdot h^{-1}\left( \frac{\delta}{2} \right) \leqslant O\left(\sqrt{\log \frac{1}{\delta}}\right).
\end{align*}
since $O(\delta)$ mass on the tails of $\cN(0,1)$ is at the tail starting at $O\left(\sqrt{\log \frac{1}{\delta}}\right)$ for $\delta$ sufficiently small.
\end{proof}

%% file: preliminaries.tex
\subsection{Sum-of-Squares Proofs to Algorithms}\label{app:sos}
Let us denote by $X$ a vector or matrix of $n$ indeterminates over $\R$.
Let $p_1(X), p_2(X), \dots p_m(X), q(X) \in \R[X]$ by polynomials with real coefficients defined over these indeterminates. We use the shorthand $\Pp$ to denote the vector consisting of polynomials $p_1, p_2, \dots, p_m$. 

Now let us consider the following set $S \coloneqq \{X \in \R^n : \Pp(X) \geqslant 0 \}$, the set of all feasible solutions of $\Pp(X) \geqslant 0$. As it will become clear, we will be interested in deriving proofs of the following form 
\begin{align*}
    \forall X \in S : q(X) \geqslant 0
\end{align*}
A degree $\ell$ Sum-of-Squares proof of the statement $q(X) \geqslant 0$ is \emph{any} finite decomposition of $q$ in the following form
\begin{align*}
    q(X) = \sum_i s_i(X)^2 \cdot \bar{p}_i(X)
\end{align*}
where $\bar{p}_i(X)$ is a product of a subset of polynomials in $\{p_1, p_2, \dots, p_m\}$ and each term in the above sum has degree at most $\ell$. For the empty set from $\{p_1, p_2, \dots, p_m\}$ we define $\bar{p}_i(X) \coloneqq 1$. 

We use the following notation for Sum-of-Squares proofs.
\begin{align*}
    \left\{\Pp(X) \geqslant 0 \right\} \sststile{\ell}{X} \left\{q(X) \geqslant 0 \right\}
\end{align*}
The above notation says that from the axioms $\Pp(X) \geqslant 0$, we can derive a Sum-of-Squares proof of the statement $q(X) \geqslant 0$ of degree at most $\ell$ in indeterminates $X$. We will use the shorthand SoS to denote Sum-of-Squares.

We have the following standard fact that we require in this paper.
\begin{fact} 
If $\cA \sststile{\ell}{X} p(X) \geqslant 0$ and $\cB \sststile{\ell'}{X} q(X) \geqslant 0$ then
\begin{align*}
    \cA \cup \cB \sststile{\ell + \ell'}{X} \{p(X) + q(X) \geqslant 0\}\ \ \ \cA \cup \cB \sststile{\ell \cdot \ell'}{X} \{p(X)\cdot q(X) \geqslant 0 \}
\end{align*}
\end{fact}

The key connection between SoS proofs and efficient algorithms is due to the pseudo-expectation operator.

\begin{definition}[pseudo-expectation] A degree $\ell$ pseudo-expectation satisfying $\Pp$ is a linear functional $\pE : \R[X]_{\leqslant \ell} \to \R$ over polynomials of degree at most $\ell$ such that:
\begin{enumerate}
    \item Normalization: $\pE [1] = 1$.
    \item Satisifiability: $ \forall i \in [m] : \pE[p_i(X) \cdot s^2(X)] \geqslant 0$ for all polynomials $s$ such that degree$(p_i(X) \cdot s^2(x)) \leqslant \ell$.
    \item Non-negativity of Squares: $\pE[q^2(x)] \geqslant 0$ for all polynomials such that degree$(q^2) \leqslant \ell$.
\end{enumerate}
\end{definition}

\begin{fact}\label{fact:pseudo}
    If $\cA \sststile{\ell}{X} \{ q(X) \geqslant 0\}$ and $\pE[.]$ is a degree $\ell$ pseudo-expectation satisfying $\cA$, then $\pE[q(X)] \geqslant 0$.
\end{fact}

\begin{fact}[SoS Algorithm]\label{alg:sos}
 \cite{shor1987quadratic, parrilo2000structured, nesterov2000squared,lasserre2001new} (Informal) Given a feasible system $\Pp$ over $\R^n$ with bit-complexity at most $(m+n)^{O(1)}$, using the ellipsoid method \cite{grotschel1981ellipsoid}, a degree $\ell$ pseudo-expectation satisfying $\Pp$ can be found in time $(m + n)^{O(\ell)}$.
\end{fact}

We remark here that by using Fact \ref{alg:sos}, we can compute a degree $\ell$ pseudo-expectation satisfying $\Pp$ in time $(m + n)^{O(\ell)}$. For our purposes, we will use this $\pE[.]$ on monomials from our polynomial system (recall that $\pE[.]$ maps polynomials to reals), and prove that the real vector obtained after applying $\pE[.]$ to the monomials will be close to the true parameter.

The key idea for statistical estimation will therefore be to have a program variable act as an estimator and relate it to the true parameter using the polynomial $q(.)$ in Fact \ref{fact:pseudo}. This will allow us to then work with $\pE[.]$. We refer the reader to \cite{raghavendra2018high} for a detailed exposition on bringing the above facts together to design efficient algorithms from SoS proofs.

For our algorithms to have polynomial running time or equivalently remain efficient, we require that $(m + n)^{O(\ell)}$ remains a polynomial in $m$ and $n$. In particular we refer to an SoS proof as \emph{low-degree} if $\ell$ does not grow with $m, n$. For all our efficient algorithms $\ell$ is a constant.

We refer the reader to \cite{barak2014sum, barak2016proofs} for a thorough overview of SoS, pseudo-expectations and their parent object of pseudo-distributions.

\subsection{Sum-of-Squares Toolkit}\label{app:sostool}
In this section we state the different SoS proofs we used in the paper. We also provide the proofs for SoS Cancellation and SoS Square Root for completeness. We note that the other proofs are well known - See \cite{barak2014sum, barak2016proofs}.\\
\begin{itemize}
    \item \textbf{Squares are non-negative}
\begin{align*}
    &\sststile{2}{a,b} 2 \cdot a \cdot b \leqslant  \left(a^2 + b^2\right) 
\end{align*}
\item \textbf{SoS Cauchy Schwarz}\\
For vectors of variables $a = (a_1, a_2 \dots a_n)$ and $b = (b_1, b_2, \dots, b_n)$, 
\begin{align*}
    \sststile{4}{a,b} \langle a, b \rangle^2 \leqslant \Vert a \Vert^2 \cdot \Vert b \Vert^2
\end{align*}
\item \textbf{SoS \Hoelder Inequality} \cite{hopkins2018mixture}\\
Let $k$ be a power of 2. For vectors of variables $w = (w_1, w_2, \dots, w_n)$ and $b = (b_1, b_2, \dots, b_n)$, 
\begin{align*}
    \left\{ \forall i \in [n]: w_i^2 = w_i \right\} \sststile{O(k)}{w, b} \left(\sum\limits_{i=1}^n w_i b_i \right)^k \leqslant \left(\sum\limits_{i=1}^n w_i \right)^{k-1} \cdot \left(\sum\limits_{i=1}^n b_i^k \right)
\end{align*}

\item \textbf{SoS AM-GM Inequality} \cite{barak2015dictionary}\\
Let $w_1, w_2, w_3, \dots, w_t$ be SoS polynomials. Then we have
\begin{align*}
    \sststile{t}{w_1, w_2, \dots, w_t} \prod_{i=1}^t w_i \leqslant \frac{1}{t}\sum_{i=1}^t w_i^t
\end{align*}
\item \textbf{SoS Triangle Inequality} \cite{kothari2018robust}\\
For even $t$
\begin{align*}
    \sststile{t}{a, b} (a+b)^t \leqslant 2^{t-1} \cdot \left( a^t + b^t \right)
\end{align*}
\item \textbf{SoS Cancellation}
For every $C > 0$,
\begin{align*}
   \left\{ X^2 \leqslant C \cdot X \right\} \sststile{2}{X} \left\{X \leqslant C \right\}
\end{align*}
\begin{proof}
\begin{align*}
        X = \frac{X}{\sqrt{C}} \cdot \sqrt{C} \leqslant \frac{1}{2} \cdot \left(\frac{X^2}{C} + C \right) \leqslant \frac{X}{2} + \frac{C}{2}
\end{align*}
where the first inequality is due to the fact that Squares are non-negative and the second inequality is because of the axioms.
Rearranging, and cancelling a factor of $\frac{1}{2}$, we get the desired result. 
\end{proof}
\item \textbf{SoS Square Root}\\
Let $k$ be a power of 2. Then for every $C > 0$,
\begin{align*}
   \left\{ X^k \leqslant C \right\} \sststile{k}{X} \left\{X \leqslant C^{1/k} \right\}
\end{align*}
\begin{proof}
Following a similar approach as the proof of SoS cancellation, we have
\begin{align*}
    X^{k/2} = \frac{X^{k/2}}{C^{1/4}} \cdot C^{1/4} \leqslant \frac{1}{2} \cdot \left(\frac{X^k}{\sqrt{C}} + \sqrt{C} \right) \leqslant \sqrt{C}
\end{align*}
Now replacing $X^{k/2}$ by $X^{k/4}$ and $C^{1/4}$ by $C^{1/8}$, and using the above derivation that $X^{k/2} \leqslant \sqrt{C}$, we can derive in SoS that $X^{k/4} \leqslant C^{1/4}$. Repeating this procedure $\log_2 k$ many times will give us the result inductively. 
\end{proof}
\end{itemize}